\documentclass[11pt,a4paper,oneside]{article}

\usepackage[utf8]{inputenc}
\usepackage[T1]{fontenc}
\usepackage{lmodern}
\usepackage[provide=*,english]{babel}
\usepackage[a4paper,margin=1in]{geometry}
\IfFileExists{setspace.sty}{\usepackage{setspace}\onehalfspacing}{}
\IfFileExists{parskip.sty}{\usepackage{parskip}}{}
\usepackage{amsmath,amssymb,amsthm,bm}
\IfFileExists{mathtools.sty}{\usepackage{mathtools}}{}
\IfFileExists{booktabs.sty}{\usepackage{booktabs}}{}
\IfFileExists{longtable.sty}{\usepackage{longtable}}{}
\IfFileExists{pdflscape.sty}{\usepackage{pdflscape}}{}
\numberwithin{equation}{section}
\IfFileExists{natbib.sty}{\usepackage[round,authoryear]{natbib}}{}
\usepackage{xcolor}
\definecolor{linkblue}{HTML}{1A73E8}
\usepackage[
 colorlinks=true,
 linkcolor=linkblue,
 citecolor=linkblue,
 urlcolor=linkblue
]{hyperref}
\IfFileExists{cleveref.sty}{%
\usepackage[capitalise,nameinlink]{cleveref}
\crefname{equation}{Eq.}{Eqs.}
\crefname{section}{Section}{Sections}
\crefname{subsection}{Section}{Sections}
\crefname{assumption}{Assumption}{Assumptions}
\crefname{theorem}{Theorem}{Theorems}
\crefname{proposition}{Proposition}{Propositions}
\crefname{lemma}{Lemma}{Lemmas}
\crefname{corollary}{Corollary}{Corollaries}
\crefname{definition}{Definition}{Definitions}
}{}
\providecommand{\cref}[1]{\ref{#1}}
\providecommand{\Cref}[1]{\ref{#1}}
\newcommand{\headingcref}[2]{\texorpdfstring{\cref{#1}}{#2}}

\IfFileExists{aliascnt.sty}{\usepackage{aliascnt}}{}

\theoremstyle{plain}
\newtheorem{theorem}{Theorem}[section]

\newaliascnt{proposition}{theorem}
\newtheorem{proposition}[proposition]{Proposition}
\aliascntresetthe{proposition}

\newaliascnt{lemma}{theorem}
\newtheorem{lemma}[lemma]{Lemma}
\aliascntresetthe{lemma}

\newaliascnt{corollary}{theorem}
\newtheorem{corollary}[corollary]{Corollary}
\aliascntresetthe{corollary}

\theoremstyle{definition}
\newaliascnt{definition}{theorem}
\newtheorem{definition}[definition]{Definition}
\aliascntresetthe{definition}

\newaliascnt{assumption}{theorem}
\newtheorem{assumption}[assumption]{Assumption}
\aliascntresetthe{assumption}

\theoremstyle{remark}
\newaliascnt{remark}{theorem}

\aliascntresetthe{remark}

\newcommand{\R}{\mathbb{R}}
\newcommand{\E}{\mathbb{E}}
\newcommand{\Var}{\mathrm{Var}}
\newcommand{\Prb}{\mathbb{P}}
\newcommand{\ind}{\mathbf{1}}
\newcommand{\expit}{\operatorname{expit}}
\newcommand{\Law}{\mathcal{L}}

\title{Win-Ratio Regression for Prioritized Composite Outcomes in
Observational Studies: Doubly Robust and Efficient Estimation with Future-Score
Correction}
\author{%
  Zhuochao Huang\thanks{Department of Statistics, University of Florida.
    Corresponding author: Zhuochao Huang,
    \texttt{zhuochao.huang@ufl.edu}}%
  \and Lucy Shao\thanks{Division of Biostatistics and Bioinformatics,
    Herbert Wertheim School of Public Health and Human Longevity Science,
    University of California San Diego.}%
  \and Yi Guo\thanks{Department of Health Outcomes and Biomedical Informatics,
    College of Medicine, University of Florida.}%
  \and Ying Lu\thanks{Department of Biomedical Data Science, Stanford University School of Medicine.}%
  \and Xin M. Tu\footnotemark[2]%
  \and Changyong Feng\thanks{Department of Biostatistics and Computational Biology,
    University of Rochester Medical Center.}%
  \and Tuo Lin\thanks{Department of Biostatistics, University of Florida.
    Corresponding author: Tuo Lin,
    \texttt{tuolin@ufl.edu}}%
}
\date{}

\begin{document}
\maketitle

\begin{abstract}

Prioritized pairwise outcomes are useful when clinical events follow a natural hierarchy, but censoring before pair resolution complicates estimation. We develop a win-ratio regression framework for this setting by defining a complete-data target over follow-up and deriving an estimating equation for the observed data. The central idea is future-score correction (FC): when censoring prevents later pairwise comparisons from being observed, the method replaces the remaining score with its conditional expectation given the observed history. This correction recovers pairwise information beyond that provided by inverse censoring weights alone. Additionally, we incorporate treatment weighting and baseline outcome augmentation to address baseline confounding. Together, these components yield double robustness for treatment assignment and censoring. Inference is obtained from U-statistic theory. Under standard regularity conditions, the AIPW-FC estimator is asymptotically normal and efficient when all nuisance functions are correctly specified. Simulations with 30\%, 50\%, and 65\% censoring show that efficiency gains from future-score correction increase with the censoring rate, with relative efficiency reaching 1.50 under 65\% censoring and near-nominal coverage for AIPW-FC. An application to OneFlorida electronic health record data illustrates the method for a composite outcome that prioritizes death over hospitalization.

\end{abstract}

\section{Introduction}
\label{sec:introduction}

Composite outcomes are common in clinical studies that evaluate treatment benefit across multiple event types. These event outcomes, including death, recurrence, and hospitalization, often differ in clinical importance and motivate analyses leveraging a prespecified clinical hierarchy. Building on the pairwise-comparison approach of Finkelstein and Schoenfeld and the generalized pairwise-comparison framework of Buyse \citep{finkelstein1999combining,buyse2010generalized}, the win-ratio method determines each pair's win status by comparing component outcomes in a prespecified hierarchy \citep{pocock2012win}. Subsequent work developed inference for the win ratio \citep{luo2015alternative,bebu2016large}, broadened the family of win statistics to include metrics such as net benefit and win odds \citep{brunner2021win,dong2023win}, and applied them in trial design \citep{pocock2024win,barnhart2025trial}.

A win estimand is determined not only by the component hierarchy, but also by the rule for ties and by the time frame over which a pair is compared. When the comparison window is left implicit, the resulting win ratio can depend on censoring and realized follow-up in addition to the underlying outcome process \citep{oakes2016win,mao2024defining,li2024elusiveness}. Under right censoring, however, the win, loss, or tie status of a pair may not be observable by the chosen horizon. A lower-priority outcome can be used only after the higher-priority outcome has been observed to be tied; if censoring occurs before that tie status is known, the pair remains unresolved even when a lower-priority outcome would have favored one subject.

Existing censoring-adjusted win-statistic methods have made important progress. Inverse-probability-of-censoring weighting (IPCW) can recover restricted win probabilities under independent or covariate-dependent censoring \citep{dong2020inverse,dong2021adjusting,cui2025ipcw,cao2026generalized}. Recent methods for missing or partially observed hierarchical outcomes further show that unresolved comparisons can still contribute information through conditional probabilities or augmented weighting \citep{liu2026estimation,liu2026estimationtwo,li2026probabilistic,fang2026improving}. These developments are primarily framed around fixed-horizon win summaries or endpoint-level missingness. In longitudinal follow-up, however, the scientific question may not be limited to who wins at a prespecified horizon. It may also concern how the pairwise comparison behaves over time: whether the pair is still unresolved, whether one subject is currently favored under the priority rule, and how long that status persists before later outcomes update the comparison. For a target that accumulates each pair's evolving comparison status over follow-up, censoring removes not just a final outcome but the future portion of the pairwise comparison that would otherwise continue to contribute information. Standard inverse weighting identifies the complete-data score under the usual assumptions, but it does not use the pair history observed before censoring, which can inform that unobserved future portion.

Another line of work extends win statistics to regression. Existing work includes proportional win-fractions regression and stratified extensions \citep{mao2021class,wang2022stratified}, as well as win-odds and generalized win-regression models \citep{song2023winodds,wang2026generalized,cao2026generalized}. Recent work on noncollapsibility and standardization further clarifies that covariate adjustment can change the interpretation of win-statistic summaries \citep{dong2026wincollapsibility}. These developments provide a natural starting point for applying win-regression methods in observational settings, including studies based on electronic health record data, where treatment is not randomized and baseline confounding must be addressed.

We propose a win-ratio regression framework centered on future-score correction. We construct the regression by accumulating each pair's comparison status over follow-up, so that a resolved comparison contributes for as long as it remains resolved. When the proportional win-fractions relationship holds at every time point, the proposed regression yields the same coefficient as proportional win-fractions regression \citep{mao2021class}; otherwise, it summarizes these pairwise comparisons over follow-up, as defined formally in Section~\ref{sec:complete-data-target}. The observed-data score replaces the unobserved future portion of the pairwise score after censoring by its conditional expectation given the observed pair histories. To support the observational EHR setting, we additionally incorporate treatment weighting and augmentation, following the logic of augmented inverse probability weighting \citep{robins1994estimation,bang2005doubly}. The resulting estimating equation achieves double robustness for both treatment assignment and censoring under the stated conditions.

We develop large-sample theory for the resulting ordered-pair U-statistic with estimated nuisance functions and derive Wald-type inference for the regression coefficients. Simulations quantify the information recovered by future-score correction as censoring increases, with relative efficiency reaching 1.50 in the heaviest-censoring scenario. An application to OneFlorida electronic health record data illustrates the method in a breast cancer cohort, comparing adjuvant chemotherapy groups with a prioritized outcome that ranks death before hospitalization.

\section{Complete-Data Estimand}
\label{sec:complete-data-target}

Win statistics are built from pairwise comparisons of a prioritized composite outcome, whose components are ordered from most to least clinically important \citep{hongyue2017win}. At any given time during follow-up, the two subjects are first compared on the most important outcome. If neither subject has a more favorable result, the next outcome in that order is considered. On the first outcome for which one subject has a more favorable result, that subject wins the comparison. If neither subject has a more favorable result on any outcome, the pair remains unresolved.

For example, when death is prioritized over hospitalization, a subject who is alive is favored over one who has died. If the death comparison is tied because both subjects are alive, their hospitalization histories are then compared. If one subject has been hospitalized by that time and the other has not, the subject who has not been hospitalized wins the comparison; if both or neither have been hospitalized, the comparison remains tied. As follow-up continues, later events can change which subject is favored. We record these changes on a discrete analysis grid, which may coincide with clinical assessment times or follow a prespecified schedule. At each grid point, the comparison uses the complete outcome histories accrued up to that time, so events between grid points enter the comparison at the next grid point.

\subsection{Observed Data Structure}

Fix an analysis window $[0,\tau]$ together with the recorded visit times
\[
0=t_0<t_1<\cdots<t_M=\tau,
\qquad
\Delta t_\ell=t_\ell-t_{\ell-1},
\qquad
\ell=1,\ldots,M.
\]
The intervals between visits may vary due to the observational clinical study design. For subject $i$, let $A_i\in\{0,1\}$ denote treatment assignment and let $X_i\in\R^p$ denote the baseline covariate vector. Let $L_{i,\ell}$ denote the vector of longitudinal covariates available at time $t_\ell$.

For treatment level $a\in\{0,1\}$, let
\[
T_{iv}(a),\qquad v=1,\ldots,V,\qquad C_i(a)
\]
denote the potential event times for the $V$ prioritized outcome components and the potential censoring time. The observed event times and censoring times are
\[
T_{iv}=T_{iv}(A_i),\qquad v=1,\ldots,V,\qquad
C_i=C_i(A_i).
\]
Baseline censoring is absent, $C_i>0$ almost surely. We record censoring on the analysis grid: a censoring time assigned to $t_\ell$ removes the subject from the risk set at the start of interval $\ell+1$, after the interval-$\ell$ contribution has been ascertained. Throughout this paper, we assume that once a patient is censored, subsequent observation of all components of that patient's prioritized outcome ceases at the same time. If different outcome components can instead have different censoring times, for example, if death continues to be recorded in registry data after cancer recurrence is no longer captured in pathology or imaging records, the current framework would need to be adjusted accordingly.

At the start of interval $\ell$, define the observed subject history
\begin{equation}
\mathcal H_{i,\ell}
:=
\left(
\begin{array}{l}
A_i,\,
X_i,\,
L_{i,0},\ldots,L_{i,\ell-1},
\\[0.25em]
\ind\{T_{iv}\le t_r\}:v=1,\ldots,V,\,
0\le r\le \ell-1,
\\[0.25em]
\ind\{C_i\le t_r\}:0\le r\le \ell-1
\end{array}
\right).
\end{equation}

\subsection{Complete-Data Score}

Consider an ordered pair of subjects, indexed by $1$ and $2$, and write
\[
\Delta X=X_1-X_2,
\qquad
\widetilde X=(1,\Delta X^{\top})^\top.
\]
Write the corresponding coefficient vector as $\beta=(\alpha,\gamma^\top)^\top\in\mathbb R^{p+1}$, so that $\beta^\top\widetilde X=\alpha+\gamma^\top\Delta X$. For ordered treated-versus-control pairs, with the treated subject listed first, the intercept coefficient $\alpha$ summarizes the overall tendency of the treated subject to be favored by the prioritized comparison; $\gamma$ describes how this tendency varies with baseline covariate differences between the two subjects. At each time $t$, the hierarchical comparison proceeds through the ordered components: the pair is first compared on the most important component and moves to the next component only if that comparison is tied. Let
\[
W^{\circ}(t)\in\{0,1\},
\qquad
R^{\circ}(t)\in\{0,1\},
\qquad
t\in[0,\tau],
\]
denote the complete-data win and determinacy processes for the ordered pair. At time $t$, $R^{\circ}(t)$ indicates whether the outcome histories accumulated through $t$ identify a winner under the priority rule, and $W^{\circ}(t)$ indicates whether the identified winner is subject $1$. When no winner is identified, $W^{\circ}(t)$ is set to zero. In intermittent clinical follow-up, these processes are piecewise constant between recorded visits:
\[
W^{\circ}(t)=W^{\circ}(t_{\ell-1}),
\qquad
R^{\circ}(t)=R^{\circ}(t_{\ell-1}),
\qquad
t\in[t_{\ell-1},t_\ell),
\]
with the terminal value at $t=\tau$ defined by right continuity. A superscript $\circ$ denotes the complete-data version of a quantity, namely the value that would be available in the absence of censoring. At any time $t$ for which the priority rule identifies a winner, the pair contributes a binary comparison: the winner is either subject $1$ or subject $2$. We model the chance that the winner is subject $1$ by a logistic function of the baseline covariate contrast,
\[
\expit\!\bigl\{\beta^\top\widetilde X\bigr\}.
\]
The coefficient vector $\beta$ is constant throughout follow-up, so that baseline covariate differences have a common interpretation across the evolving prioritized comparisons; this is also the proportional win-fractions structure underlying \citep{mao2021class,wang2022stratified}. Equivalently,
\[
\Prb\{W^{\circ}(t)=1
\mid R^{\circ}(t)=1, X_1, X_2\}
=
\expit\!\bigl\{\beta^\top\widetilde X\bigr\}.
\]
This leads to the complete-data residual process
\begin{equation}
M^{\circ}(t;\beta)
:=
W^{\circ}(t)
-
R^{\circ}(t)\expit\!\bigl\{\beta^\top\widetilde X\bigr\}.
\end{equation}
At the start of interval $\ell$, write $M^{\circ}(t_{\ell-1};\beta)$ for the corresponding residual value. With comparisons updated at grid points, $M^{\circ}(t;\beta)$ remains constant within each interval. The interval-$\ell$ complete-data score contribution is therefore
\[
S_{\ell}^{\circ}(\beta)
:=
\int_{t_{\ell-1}}^{t_\ell}
\widetilde X M^{\circ}(t;\beta)\,dt
=
\Delta t_\ell\,\widetilde X\,
M^{\circ}(t_{\ell-1};\beta),
\]
and the complete-data pair score across follow-up is
\[
S^{\circ}(\beta)
:=
\sum_{\ell=1}^{M} S_{\ell}^{\circ}(\beta)
=
\widetilde X\int_0^\tau M^{\circ}(t;\beta)\,dt.
\]

For a propensity score $e$, define the ordered treated-versus-control weight
\begin{equation}
\omega^{1,0}(e)
:=
\frac{A_1(1-A_2)}{e(X_1)\{1-e(X_2)\}}.
\end{equation}
Let $e_0(x)=\Prb(A=1\mid X=x)$ denote the true propensity score. The treatment-weighted complete-data score is therefore
\begin{equation}
\label{eq:full-score}
\omega^{1,0}(e_0)S^{\circ}(\beta)
=
\omega^{1,0}(e_0)
\widetilde X
\int_0^\tau M^{\circ}(t;\beta)\,dt
=
\omega^{1,0}(e_0)
\sum_{\ell=1}^{M}
\Delta t_\ell\,\widetilde X\,
M^{\circ}(t_{\ell-1};\beta).
\end{equation}
\Cref{eq:full-score} integrates the pairwise residual process over follow-up and weights each ordered treated-versus-control pair by the standard inverse-probability weight $\omega^{1,0}(e_0)$. The complete-data pair process is piecewise constant between recorded visits and the score can be written as a sum over visit intervals, with each residual contribution weighted by the interval length. Earlier resolution therefore receives greater weight, since the pair remains in the resolved state longer.

\begin{definition}[Time-Constant Logit Estimand]
\label{def:target}
The time-constant logit target $\beta_0\in\R^{p+1}$ satisfies
\begin{equation}
\label{eq:target}
\E\!\left\{\omega^{1,0}(e_0)S^{\circ}(\beta_0)\right\}=0.
\end{equation}
\end{definition}

Thus, $\beta_0$ summarizes how baseline covariate differences relate to whether the treated subject is favored at times when a winner is identified. Each resolved comparison contributes for as long as the winner remains identified. If the proportional win-fractions relationship holds with the same coefficient at every time point, the proposed regression yields the same coefficient as proportional win-fractions regression \citep{mao2021class,wang2022stratified}; otherwise, $\beta_0$ summarizes the treated subject's tendency to be favored during the portions of follow-up in which a winner is identified. \Cref{app:complete-data-estimand-properties} gives conditions under which \cref{eq:target} has a unique root and under which the root equals the coefficient obtained from proportional win-fractions regression.

\section{Observed-Data Scores}
\label{sec:observed-data-representation}
Definition~\ref{def:target} describes the coefficient that would be obtained from complete follow-up of each ordered pair. With right censoring, some future interval contributions to $S^{\circ}(\beta)$ are not observed. Under sequentially independent censoring, these contributions can be recovered in expectation from the observed pair history. We use this future-score correction, together with augmented inverse-probability weighting for baseline confounding, to construct an observed-data score with the same population root $\beta_0$.

\subsection{Censoring-Corrected Score}
\label{sec:full-correction}

For $k\in\{1,2\}$, let
\[
Y_{k,\ell}=\ind\{C_k>t_\ell\}
\]
under the interval convention in Section~\ref{sec:complete-data-target}, and define the pair at-risk indicator
\[
Y_\ell=Y_{1,\ell}Y_{2,\ell}.
\]
Define the true subject-level discrete censoring hazard by
\[
  \lambda_{k,\ell}^C
  =
  \Prb\{C_k\in(t_{\ell-1},t_\ell]\mid C_k>t_{\ell-1},
  \mathcal H_{k,\ell}^{+}\},
\]
where $\mathcal H_{k,\ell}^{+}$ is the subject history through interval $\ell$, after the interval-$\ell$ contribution has been ascertained and immediately before censoring at $t_\ell$. Define the pair history by
\[
\mathcal H_\ell^+
=
\sigma(\mathcal H_{1,\ell}^{+},\mathcal H_{2,\ell}^{+}).
\]
Set $G_{0,0}=1$ and write the true pair-level censoring survival probability through $t_\ell$ as
\[
  G_{0,\ell}
  =
  \prod_{r=1}^{\ell}
  \{1-\lambda_{1,r}^C\}\{1-\lambda_{2,r}^C\}.
\]
Write $G_0=(G_{0,\ell})_{\ell=0}^{M}$. The subject-level hazards induce the true pair-level censoring hazard
\[
  \lambda_\ell^C
  =
  1-\{1-\lambda_{1,\ell}^C\}\{1-\lambda_{2,\ell}^C\}.
\]
For a working pair censoring survival $G$, let $G_\ell$ denote its value through $t_\ell$, with baseline survival equal to one, and define
\[
\lambda_\ell^C(G)
=
1-\frac{G_\ell}{G_{\ell-1}}.
\]
Let
\[
  \Delta N_{\ell}^{C}
  =
  Y_{\ell-1}\ind\{Y_\ell=0\}
\]
be the censoring jump for the pair in interval $\ell$. The corresponding pairwise censoring martingale increment is
\[
  dM_{\ell}^C(G)
  =
  \Delta N_{\ell}^{C}
  -
  Y_{\ell-1}\lambda_\ell^C(G),
\]
which is zero after the pair has left the uncensored risk set.

The contribution remaining after interval $\ell$ is
\[
\Psi_\ell(\beta)
:=
\sum_{r>\ell}S_{r}^{\circ}(\beta).
\]
The true future-score projection is
\[
  \Phi_{0,\ell}(\beta)
  :=
  \E\{\Psi_\ell(\beta)\mid \mathcal H_\ell^{+}\}.
\]
See Figure~\ref{fig:future-score-mechanism} for an illustration of the mechanism of the future-score projection.

\begin{figure}[t]
\centering
\includegraphics[width=\textwidth]{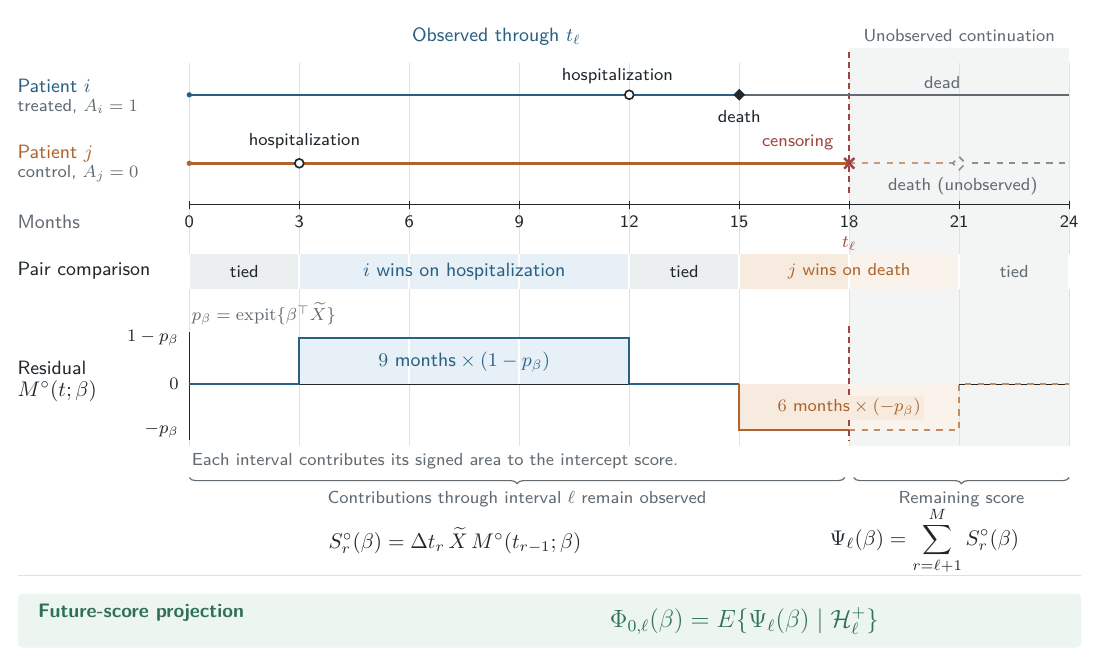}
\caption{Mechanism of the future-score projection. Censoring leaves subsequent interval contributions unobserved, and the future-score projection is the conditional expectation of this remaining score given the observed pair history.}
\label{fig:future-score-mechanism}
\end{figure}

Write $\Phi_0(\beta)=\{\Phi_{0,\ell}(\beta):\ell=1,\ldots,M-1\}$. For a square-integrable predictable working projection $\Phi$, write its interval-$\ell$ value as $\Phi_\ell(\beta)$. If censoring is recorded at the end of interval $\ell$, the pair still contributes the interval-$\ell$ score, but is no longer observed for later intervals. Thus the inverse-weighted interval score $S_{\ell}^{\circ}(\beta)$ uses $G_{\ell-1}$, the probability of remaining uncensored through $t_{\ell-1}$, whereas the correction term involving $\Phi_{\ell}(\beta)$ uses $G_{\ell}$ and predicts only the remaining score from intervals $r>\ell$. The inverse-censoring-weighted score is
\[
  S^{\mathrm{IPW}}(\beta;G)
  =
  \sum_{\ell=1}^{M}
  \frac{Y_{\ell-1}}{G_{\ell-1}}
  S_{\ell}^{\circ}(\beta).
\]
The censoring-corrected score for the pair is
\begin{equation}
\label{eq:full-censor-score}
  S^{\mathrm{FC}}(\beta;G,\Phi)
  =
  S^{\mathrm{IPW}}(\beta;G)
  +
  \sum_{\ell=1}^{M-1}
  \frac{Y_{\ell-1}}{G_{\ell}}
  \Phi_{\ell}(\beta)\,dM_{\ell}^{C}(G).
\end{equation}
The first term is standard pairwise IPCW. The second follows from the inverse-censoring telescoping identity: when a censoring jump removes a pair from future risk, the predictable projection replaces the remaining complete-data score that would otherwise be lost.

\begin{theorem}[Censoring Double Robustness]
\label{thm:future-score-dr}
Suppose Assumptions~\ref{ass:consistency}--\ref{ass:second-moment} hold. If $G=G_0$, then the score in \eqref{eq:full-censor-score} satisfies
\begin{equation}
\label{eq:censoring-correct-G}
\E\{S^{\mathrm{FC}}(\beta;G,\Phi)\}
=
\E\{S^{\circ}(\beta)\}
\end{equation}
for every square-integrable predictable function $\Phi$. If $\Phi=\Phi_0$, the true future-score projection, then
\[
  \E\{S^{\mathrm{FC}}(\beta;G,\Phi_0)\}
  =
  \E\{S^{\circ}(\beta)\}
\]
for every predictable censoring survival $G$ bounded away from zero. Thus the censoring correction is consistent if either the censoring survival is correct or the future-score projection is correct.
\end{theorem}

\subsection{Treatment-Augmented Score}
\label{sec:treatment-aipw}

Define the expected complete-data score of a treated-versus-control pair with baseline covariates $(x_1,x_0)$ as
\[
  m_0(x_1,x_0;\beta)
  =
  \E\{S^{\circ}(\beta)\mid X_1=x_1,X_2=x_0,A_1=1,A_2=0\}.
\]
Write $m_0(\beta)=m_0(\cdot,\cdot;\beta)$. The regression $m_0$, which captures how this expected contribution varies with baseline covariates, is used to augment inverse-probability treatment weighting in the estimating equation. For a nuisance vector $\eta=(e,G,\Phi,m)$, define the kernel
\begin{equation}
\label{eq:aipw-kernel}
  K(\beta;\eta)
  =
  \omega^{1,0}(e)
  \{S^{\mathrm{FC}}(\beta;G,\Phi)-m(X_1,X_2;\beta)\}
  +
  m(X_1,X_2;\beta).
\end{equation}
Set $\eta_0=\{e_0,G_0,\Phi_0(\beta_0),m_0(\beta_0)\}$.

\begin{theorem}[Double Robustness]
\label{thm:aipw-dr}
Suppose Assumptions~\ref{ass:consistency}--\ref{ass:treatment} hold and the nuisance functions are square-integrable, with $\Phi_\ell$ being $\mathcal H_\ell^{+}$-measurable and $e$, $1-e$, and $G$ bounded away from zero. Then
\[
  \E\{K(\beta_0;\eta)\}=0
\]
whenever either $e=e_0$ or $m=m_0$, and either $G=G_0$ or $\Phi=\Phi_0$. At $\eta_0$, the Gateaux derivative of $\E\{K(\beta_0;\eta)\}$ with respect to each nuisance component is zero. Consequently, nuisance estimation contributes only second-order drift:
\[
  \E\{K(\beta_0;\eta)-K(\beta_0;\eta_0)\}
  =
  O\!\left(\|e-e_0\|\,\|m-m_0\|
  + \|G-G_0\|\,\|\Phi-\Phi_0\|\right).
\]
All norms are $L_2$ norms induced by $P_0$.
\end{theorem}

With $\eta=\eta_0$, the theorem gives
\[
\E\{K(\beta_0;\eta_0)\}=0,
\]
while Definition~\ref{def:target} gives
\[
\E\{\omega^{1,0}(e_0)S^\circ(\beta_0)\}=0.
\]
Thus the observed-data score $K$ retains the complete-data target $\beta_0$ as a population root: the same value $\beta_0$ satisfies both population estimating equations.

\subsection{Future-Score Projection}
\label{sec:future-projection}

By definition, the true future-score projection $\Phi_{0,\ell}(\beta)$ is the conditional mean of the remaining complete-data score $\Psi_\ell(\beta)$ given the observed pair history $\mathcal H_\ell^{+}$ immediately before censoring. It can be estimated either from a regression that directly targets this conditional mean or by modeling each subject's future trajectory from the observed history and using the resulting models to construct the future-score projection. The development in this section uses the construction based on modeling future trajectories, while the direct regression approach is discussed in Appendix~\ref{app:asymptotic-linearity}. We use $\widehat\Phi$ to denote an estimator and add a corresponding superscript when the construction must be distinguished.

For subject $k\in\{1,2\}$, let $\mathcal U_{k,\ell}$ denote the complete future trajectory on $\{t_\ell,\ldots,t_M\}$, including future event indicators and longitudinal covariate values. Define the conditional distribution of $\mathcal U_{k,\ell}$ given $\mathcal H_{k,\ell}^{+}$ as
\[
  \Law\!\left(
  \mathcal U_{k,\ell}
  \middle|
  \mathcal H_{k,\ell}^{+}
  \right).
\]

Let $\Psi_\ell(\beta;u_1,u_2)$ denote the remaining score obtained when the future trajectories of subjects $1$ and $2$ are set to $(u_1,u_2)$. Because subjects are sampled independently, the two future trajectories are conditionally independent given their respective observed histories. Therefore,
\begin{equation}
\label{eq:Phi-integral}
  \Phi_{0,\ell}(\beta)
  =
  \int
  \Psi_\ell(\beta;u_1,u_2)\,
  d\Law(u_1\mid\mathcal H_{1,\ell}^{+})\,
  d\Law(u_2\mid\mathcal H_{2,\ell}^{+}).
\end{equation}

To evaluate \eqref{eq:Phi-integral}, we estimate each conditional future-trajectory distribution using fitted one-step transition models. Let
$\widehat Q_{a,r}(\cdot\mid \mathcal H_{k,r-1}^{+})$
denote the fitted one-step transition distribution of the interval-$r$ variables under treatment $a$ conditional on the subject history.

Applying these one-step transitions successively yields the fitted conditional distribution of $\mathcal U_{k,\ell}$ given $\mathcal H_{k,\ell}^{+}$:
\[
  \widehat{\Law}\!\left(
  \mathcal U_{k,\ell}
  \middle|
  \mathcal H_{k,\ell}^{+}
  \right)
  =
  \prod_{r=\ell+1}^{M}
  \widehat Q_{A_k,r}(\cdot\mid \mathcal H_{k,r-1}^{+}).
\]
Starting from $\mathcal H_{k,\ell}^{+}$, the future interval variables are generated sequentially from $\widehat Q_{A_k,r}$, with the conditioning history at interval $r$ consisting of $\mathcal H_{k,\ell}^{+}$ and all variables generated at earlier future intervals. The fitted future-score projection is
\begin{equation}
  \widehat\Phi_{\ell}^{\mathrm{fit}}(\beta)
  =
  \int
  \Psi_\ell(\beta;u_1,u_2)\,
  d\widehat{\Law}(u_1\mid\mathcal H_{1,\ell}^{+})\,
  d\widehat{\Law}(u_2\mid\mathcal H_{2,\ell}^{+}).
\end{equation}
When the fitted future-trajectory distributions have finite supports $\mathcal S_{1,\ell}$ and $\mathcal S_{2,\ell}$, with probabilities $\widehat p_{k,\ell}$ induced by the forward recursion, this projection is
\begin{equation}
\label{eq:recursive-projection}
  \widehat\Phi_{\ell}^{\mathrm{det}}(\beta)
  =
  \sum_{u_1\in\mathcal S_{1,\ell}}
  \sum_{u_2\in\mathcal S_{2,\ell}}
  \Psi_\ell(\beta;u_1,u_2)\,
  \widehat p_{1,\ell}(u_1\mid \mathcal H_{1,\ell}^{+})\,
  \widehat p_{2,\ell}(u_2\mid \mathcal H_{2,\ell}^{+}).
\end{equation}
Figure~\ref{fig:future-score-recursion} illustrates the recursive construction in \eqref{eq:recursive-projection} with death and hospitalization as the prioritized outcomes.

\begin{figure}[!htbp]
\centering
\includegraphics[width=\textwidth]{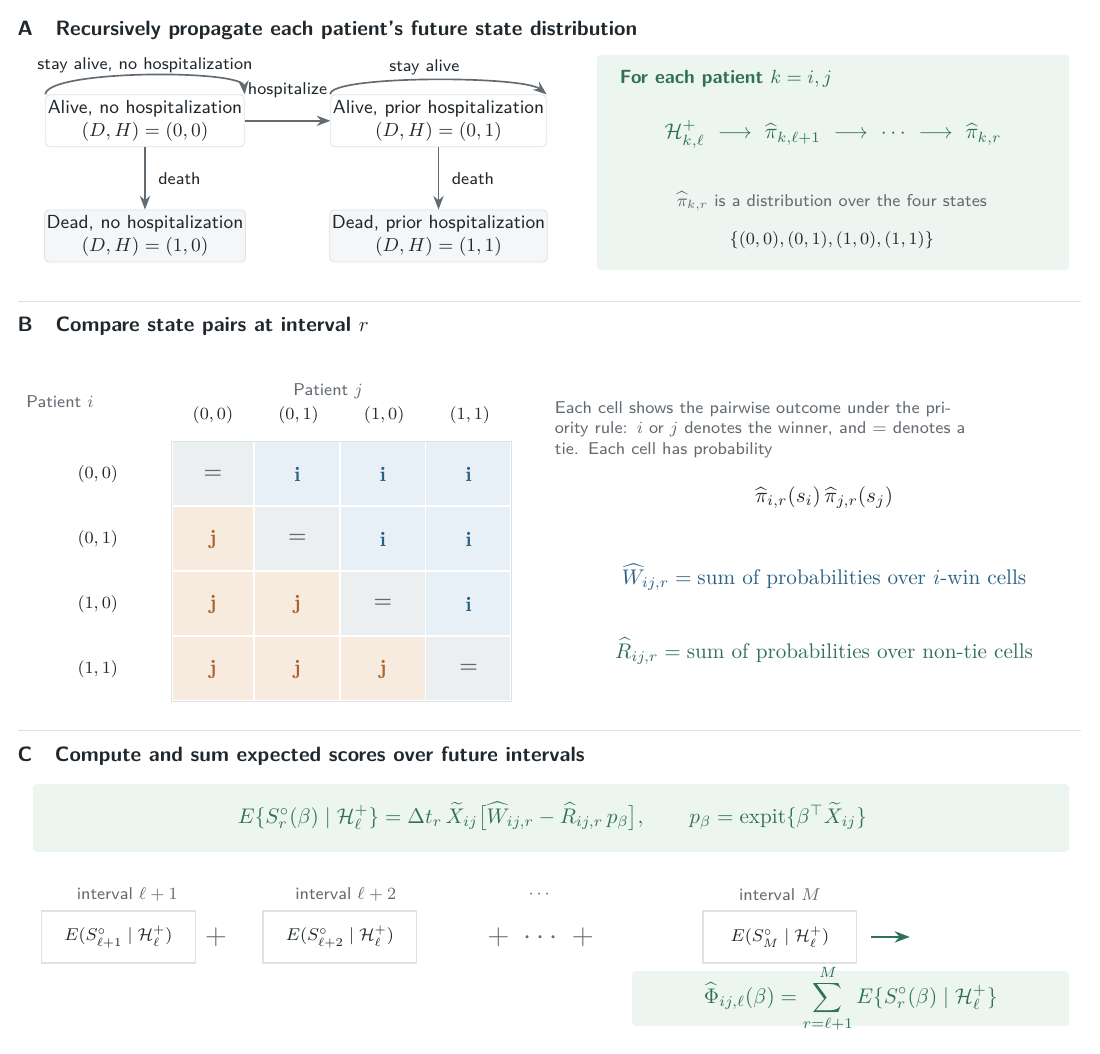}
\caption{Recursive construction of the future-score projection. Fitted transition models propagate each patient's future state distribution, and the resulting distributions are combined to obtain win and resolution probabilities at each interval and converted into expected scores that are summed over future intervals.}
\label{fig:future-score-recursion}
\end{figure}

For continuous or large future state spaces, we can approximate the fitted projection using the Monte Carlo approximation based on $B$ independent draws $(\widehat{\mathcal U}_{1,\ell}^{(b)}, \widehat{\mathcal U}_{2,\ell}^{(b)})$, $b=1,\ldots,B$, from $\widehat{\Law}(\mathcal U_{1,\ell}\mid\mathcal H_{1,\ell}^{+}) \otimes \widehat{\Law}(\mathcal U_{2,\ell}\mid\mathcal H_{2,\ell}^{+})$:
\begin{equation}
\label{eq:mc-projection}
  \widehat\Phi_{\ell}^{B}(\beta)
  =
  \frac{1}{B}
  \sum_{b=1}^{B}
  \Psi_\ell
  \!\left(
  \beta;
  \widehat{\mathcal U}_{1,\ell}^{(b)},
  \widehat{\mathcal U}_{2,\ell}^{(b)}
  \right).
\end{equation}

\section{U-Statistic Inference}
\label{sec:ustat-inference}

We estimate $\beta_0$ by solving the following U-statistic estimating equation \citep{kowalski2008modern}:
\begin{equation}
  \widehat U_n(\beta)
  =
  \frac{1}{n(n-1)}
  \sum_{i\ne j}K_{ij}(\beta;\widehat\eta_{ij}),
  \qquad
  \widehat U_n(\widehat\beta)=0.
\end{equation}

Here $K_{ij}(\beta;\widehat\eta_{ij})$ is the kernel in \eqref{eq:aipw-kernel} evaluated on ordered pair $(i,j)$ using the cross-fitted nuisance estimate $\widehat\eta_{ij}$ fitted outside the folds containing $i$ and $j$. Hence neither subject contributes to the nuisance fit entering $K_{ij}$.

Let $O_i$ denote the observed data for subject $i$ over $[0,\tau]$. Let $P$ denote a distribution of $O_i$ in the nonparametric observed-data model, and let $P_0$ denote the true distribution of $O_i$. Under $P$, let $e_P$ and $G_P$ denote its treatment and censoring conditional laws. For a given observed-data law $P$, define
\[
  \Phi_{P,\ell}^{\mathrm{obs}}(\beta)
  =
  \E_{P\otimes P}\!\left[
    \left.
    \sum_{r>\ell}
    \frac{
      Y_{r-1}G_{P,\ell-1}
    }{
      G_{P,r-1}
    }
    S_{r}^{\circ}(\beta)
    \,\right|\,
    \mathcal H_{\ell}^{+}
  \right].
\]
Write $\Phi_P^{\mathrm{obs}}(\beta)=\{\Phi_{P,\ell}^{\mathrm{obs}}(\beta):\ell=1,\ldots,M-1\}$. Define
\[
  m_P^{\mathrm{obs}}(x_1,x_0;\beta)
  =
  \E_{P\otimes P}\!\left[
    S^{\mathrm{FC}}\{\beta;G_P,\Phi_P^{\mathrm{obs}}(\beta)\}
    \mid X_1=x_1,X_2=x_0,A_1=1,A_2=0
  \right].
\]
Write $m_P^{\mathrm{obs}}(\beta)=m_P^{\mathrm{obs}}(\cdot,\cdot;\beta)$, and set $\eta_P^{\mathrm{obs}}(\beta)=\{e_P,G_P,\Phi_P^{\mathrm{obs}}(\beta),m_P^{\mathrm{obs}}(\beta)\}$. For any observed-data distribution $P$, define the moment functional at fixed $\beta$ by
\begin{equation}
\label{eq:observed-moment-functional}
  \psi_\beta(P)
  =
  \E_{P\otimes P}
  \left[
  K\{\beta;\eta_P^{\mathrm{obs}}(\beta)\}
  \right].
\end{equation}
At $P_0$, $e_{P_0}=e_0$ and $G_{P_0}=G_0$, so \Cref{thm:aipw-dr} applied to $\eta_{P_0}^{\mathrm{obs}}(\beta_0)$ gives
\[
\psi_{\beta_0}(P_0)
=
\E_{P_0\otimes P_0}\!\left[
K\{\beta_0;\eta_{P_0}^{\mathrm{obs}}(\beta_0)\}
\right]
=
0.
\]

For the influence-function representation below, sequentially independent censoring also yields
\begin{equation}
\label{eq:truth-law-nuisance}
\Phi_{P_0}^{\mathrm{obs}}(\beta_0)=\Phi_0(\beta_0),
\qquad
m_{P_0}^{\mathrm{obs}}(\beta_0)=m_0(\beta_0).
\end{equation}
These identities are verified in the Appendix. Together with $e_{P_0}=e_0$ and $G_{P_0}=G_0$, they give $K\{\beta_0;\eta_{P_0}^{\mathrm{obs}}(\beta_0)\}=K(\beta_0;\eta_0)$.

At $(P_0,\beta_0)$, define the symmetrized kernel by
\[
  \bar K(O_1,O_2)
  =
  \frac{K_{12}(\beta_0;\eta_0)+K_{21}(\beta_0;\eta_0)}{2}.
\]
Define the first Hoeffding projection by
\[
  \kappa(o)
  =
  \E\{\bar K(O_1,O_2)\mid O_1=o\}.
\]
Because $\E\{\bar K(O_1,O_2)\}=0$, this projection is centered. This projection governs how perturbations of the observed-data distribution propagate through the estimating equation to the target parameter.

\begin{proposition}[Efficient Influence Function]
\label{prop:canonical-gradient}
Suppose Assumptions~\ref{ass:consistency}--\ref{ass:kernel-regularity} hold. The efficient influence function, equivalently the canonical gradient, of the fixed-$\beta_0$ moment $\psi_{\beta_0}(P)$ at $P_0$ is
\begin{equation}
  D_\psi^{\mathrm{eff}}(O)
  =
  2\kappa(O).
\end{equation}
If $\beta(P)$ is the locally unique solution of $\psi_{\beta(P)}(P)=0$, then its efficient influence function is
\begin{equation}
\label{eq:beta-eif}
  D_\beta^{\mathrm{eff}}(O)
  =
  -\left\{
    \left.
    \frac{\partial}{\partial\beta^\top}
    \psi_\beta(P_0)
    \right|_{\beta=\beta_0}
  \right\}^{-1}
  D_\psi^{\mathrm{eff}}(O)
  =
  -2\left\{
    \left.
    \frac{\partial}{\partial\beta^\top}
    \psi_\beta(P_0)
    \right|_{\beta=\beta_0}
  \right\}^{-1}
  \kappa(O).
\end{equation}
\end{proposition}

For theory to apply to the estimated $\widehat{\beta}$, we next establish that the root of the empirical estimating equation converges to the population root $\beta_0$ under mild conditions.

\begin{proposition}[Consistency]
\label{prop:estimator-consistency}
Suppose Assumptions~\ref{ass:consistency}--\ref{ass:crossfit-regularity} hold. Suppose the cross-fitted nuisance estimators are uniformly $L_2(P_0)$ consistent, and the fitted propensity scores are bounded away from zero and one with probability tending to one. Let $\widehat\beta$ satisfy
\[
  \|\widehat U_n(\widehat\beta)\|=o_p(1).
\]
Then
\[
  \widehat\beta\xrightarrow{p}\beta_0.
\]
\end{proposition}

We can then establish the following result which characterizes the asymptotic properties of $\widehat\beta$.

\begin{theorem}[Asymptotic Linearity and Normality]
\label{thm:asymptotic-linearity}
Suppose the conditions of \Cref{prop:estimator-consistency} hold. Further assume that
\[
  \|\widehat U_n(\widehat\beta)\|=o_p(n^{-1/2}),
\]
and that
\[
  \|\widehat e-e_0\|\,\|\widehat m-m_0\|=o_p(n^{-1/2}),
  \qquad
  \|\widehat G-G_0\|\,\|\widehat\Phi-\Phi_0\|=o_p(n^{-1/2}).
\]
Then
\[
  \sqrt n(\widehat\beta-\beta_0)
  =
  \frac{1}{\sqrt n}\sum_{i=1}^{n}
  D_\beta^{\mathrm{eff}}(O_i)+o_p(1),
\]
and
\[
  \sqrt n(\widehat\beta-\beta_0)
  \rightsquigarrow
  N\!\left(0,\,
  4\left\{
    \left.
    \frac{\partial}{\partial\beta^\top}
    \psi_\beta(P_0)
    \right|_{\beta=\beta_0}
  \right\}^{-1}
  \mathrm{Var}\{\kappa(O)\}
  \left\{
    \left.
    \frac{\partial}{\partial\beta^\top}
    \psi_\beta(P_0)
    \right|_{\beta=\beta_0}
  \right\}^{-\top}\right).
\]
\end{theorem}

These rate requirements are met, for example, when $\widehat e$ and $\widehat G$ are root-$n$ consistent while $\widehat m$ and $\widehat\Phi$ are consistent. The Appendix gives sufficient conditions under which the nuisance estimators used here attain these rates.

The empirical Hoeffding projection provides a variance estimator. Define
\[
  \widehat{\bar K}_{ij}
  =
  \frac{1}{2}
  \{K_{ij}(\widehat\beta;\widehat\eta_{ij})
  +K_{ji}(\widehat\beta;\widehat\eta_{ji})\},
\]
\[
  \widehat\kappa_i
  =
  \frac{1}{n-1}\sum_{j\ne i}
  \widehat{\bar K}_{ij}
  -
  \frac{1}{n}\sum_{r=1}^{n}\frac{1}{n-1}\sum_{s\ne r}
  \widehat{\bar K}_{rs}.
\]
After constructing the cross-fitted nuisance fits, set
\[
  \widehat\Sigma_\beta
  =
  4\left\{
    \left.
    \frac{\partial}{\partial\beta^\top}
    \widehat U_n(\beta)
    \right|_{\beta=\widehat\beta}
  \right\}^{-1}
  \left\{ \frac{1}{n}\sum_{i=1}^{n}
  \widehat\kappa_i\widehat\kappa_i^\top\right\}
  \left\{
    \left.
    \frac{\partial}{\partial\beta^\top}
    \widehat U_n(\beta)
    \right|_{\beta=\widehat\beta}
  \right\}^{-\top}.
\]

\begin{corollary}[Wald Inference]
\label{cor:wald-inference}
Suppose the conditions of \Cref{thm:asymptotic-linearity} hold. Then
\[
  \widehat\Sigma_\beta
  \xrightarrow{p}
  \Sigma_\beta^{\mathrm{eff}}
  :=
  \E\!\left\{
  D_\beta^{\mathrm{eff}}(O)
  D_\beta^{\mathrm{eff}}(O)^\top
  \right\}
  =
  4\left\{
    \left.
    \frac{\partial}{\partial\beta^\top}
    \psi_\beta(P_0)
    \right|_{\beta=\beta_0}
  \right\}^{-1}
  \Var\{\kappa(O)\}
  \left\{
    \left.
    \frac{\partial}{\partial\beta^\top}
    \psi_\beta(P_0)
    \right|_{\beta=\beta_0}
  \right\}^{-\top}.
\]
Consequently, $\widehat\Sigma_\beta/n$ consistently estimates the sampling covariance of $\widehat\beta$. For component $r$ with $\Sigma_{\beta,rr}^{\mathrm{eff}}>0$ and $0<\alpha<1$, the Wald interval
\[
  \widehat\beta_r
  \ \pm\
  z_{1-\alpha/2}
  \sqrt{\widehat\Sigma_{\beta,rr}/n}
\]
has limiting coverage $1-\alpha$.
\end{corollary}

The conclusions of \Cref{thm:asymptotic-linearity,cor:wald-inference} also hold when $\widehat\Phi^{\mathrm{det}}$ is replaced by the Monte Carlo approximation $\widehat\Phi^B$ defined in \eqref{eq:mc-projection}, provided the resulting Monte Carlo approximation error is $o_p(n^{-1/2})$. For Wald inference, the variance estimator uses a separate Monte Carlo approximation; see \Cref{lem:mc-error} and the proof of \Cref{cor:wald-inference}.
\begin{corollary}[Semiparametric Efficiency]
\label{cor:semiparametric-efficiency}
Under the conditions of \Cref{prop:canonical-gradient} and \Cref{thm:asymptotic-linearity}, the cross-fitted estimator based on the AIPW-FC score is regular and attains the semiparametric efficiency bound $\Sigma_\beta^{\mathrm{eff}}$. In particular, if $\widetilde\beta$ is any regular asymptotically linear estimator of $\beta_0$ in the same observed-data model, with influence function $D_{\widetilde\beta}$, then
\[
  \Var(D_{\widetilde\beta})
  -
  \Sigma_\beta^{\mathrm{eff}}
  =
  \Var\!\left\{
  D_{\widetilde\beta}-D_\beta^{\mathrm{eff}}
  \right\}
  \succeq 0.
\]
Equality holds if and only if $D_{\widetilde\beta}=D_\beta^{\mathrm{eff}}$ almost surely.
\end{corollary}

\section{Simulation Studies}
\label{sec:simulation-studies}

We conducted simulation studies to evaluate finite-sample bias, efficiency, and interval estimation for the proposed win-ratio regression estimators, with particular emphasis on the information recovered by the future-score correction. Each dataset contained 500 independent subjects. We generated a continuous baseline biomarker $X\sim\mathrm{Uniform}(-1,1)$ and a baseline cardiovascular disease (CVD) indicator $Z^{\mathrm{cvd}}\sim\mathrm{Bernoulli}(0.30)$. Treatment assignment followed
\[
  \Prb(A=1\mid X,Z^{\mathrm{cvd}})
  =
  \expit(-0.30+0.40X+1.60Z^{\mathrm{cvd}}).
\]
For the simulations, we used a discrete analysis grid with eight three-month intervals. Throughout the simulation studies and data application, prioritized comparisons used a common two-level rule: death had first priority; among pairs with the same death status, a patient free of prior hospitalization was favored over a patient with prior hospitalization; otherwise the pair remained unresolved until a later grid point. This rule follows the hierarchical win/loss construction of the win-ratio statistic \citep{pocock2012win}, applied here over follow-up. The event hazards were calibrated to yield approximately 5\% mortality and 60\% hospitalization by 24 months. Censoring depended on baseline covariates and prior hospitalization; the intercept of the discrete censoring hazard was calibrated to produce approximately 30\%, 50\%, and 65\% censoring, while all other features of the data-generating mechanism were held fixed. Thus the three scenarios represent increasing censoring intensity, with more pairwise comparisons at later grid points unobserved in the higher-censoring scenarios. The complete data-generating mechanism is described in Appendix~\ref{app:simulation-supplement}.

We compared inverse probability weighting (IPW), IPW with future-score correction (IPW-FC), augmented IPW (AIPW), and the AIPW estimator with future-score correction (AIPW-FC); \Cref{app:simulation-estimators} defines the four estimators used in the simulation tables. All working models were correctly specified for the main simulation results. We additionally examined misspecification scenarios to assess the treatment and censoring robustness mechanisms; those results are reported in \Cref{app:simulation-dr}. Treatment, censoring, death, and hospitalization models were fitted with five-fold subject-level cross-fitting. The future-score projection and baseline outcome regression were evaluated by deterministic recursion under fitted transition models on the finite state space for death and hospitalization.

For each censoring level, we generated 500 Monte Carlo datasets. We computed the complete-data target independently of censoring by deterministic numerical integration over $X$, summation over the two values of $Z^{\mathrm{cvd}}$, and numerical solution of the population score equation. We report the true value, bias, empirical standard deviation (ESD), average estimated standard error (ASE), coverage of nominal 95\% intervals, and relative efficiency (RE). For IPW-FC, RE is the squared ESD of IPW divided by that of IPW-FC; for AIPW-FC, the reference is AIPW. For all four methods, inference used 200 nonparametric bootstrap samples at the subject level, with every nuisance model and estimating equation refitted in each sample. Bootstrap ASE and coverage are reported for all methods; the model-based standard error and Wald coverage from Corollary~\ref{cor:wald-inference} are reported in addition for AIPW-FC.

Table~\ref{tab:simulation-main} summarizes performance across the three censoring levels. Future-score correction improved precision for both IPW and AIPW at each censoring level, and the gain increased as censoring became heavier. At 30\% censoring, RE ranged from 1.06 to 1.09 for IPW-FC and from 1.07 to 1.10 for AIPW-FC. At 50\% censoring, the respective ranges increased to 1.18--1.22 and 1.20--1.26. Under 65\% censoring, RE reached 1.34--1.44 for IPW-FC and 1.38--1.50 for AIPW-FC. In contrast, AIPW had ESD similar to or slightly larger than IPW. Given the modest differences and the finite sample size, this pattern may reflect additional variability from estimating the baseline outcome regression, particularly because the IPW estimator was already nearly unbiased in these scenarios. Biases were small relative to the Monte Carlo standard deviations. For AIPW-FC, the model-based ASE tracked the ESD closely, and model-based coverage remained near the nominal level across the three censoring scenarios. These results agree with the mechanism of the method: heavier censoring leaves more future score contributions unobserved, and the future-score projection recovers part of their otherwise discarded information.

\begin{landscape}
\begin{table}[!htbp]
\centering
\caption{Simulation results under 30\%, 50\%, and 65\% censoring. ESD denotes empirical
standard deviation, ASE denotes average estimated standard error, and coverage
is for nominal 95\% intervals. For AIPW-FC, ASE and coverage are shown as
model/bootstrap; for the other estimators, bootstrap values are shown. RE
compares each estimator with future-score correction with the corresponding estimator without correction.}
\label{tab:simulation-main}
\scriptsize
\setlength{\tabcolsep}{3pt}
\begin{tabular}{cllrrrrrr}
\toprule
Censoring & Component & Estimator & True & Bias & ESD
& ASE & Coverage & RE \\
\midrule
30\% & Intercept & IPW & 0.58 & 0.01 & 0.18 & 0.18 & 0.95 & -- \\
     & & AIPW & 0.58 & 0.01 & 0.18 & 0.19 & 0.95 & -- \\
     & & IPW-FC & 0.58 & 0.02 & 0.18 & 0.18 & 0.95 & 1.06 \\
     & & AIPW-FC & 0.58 & 0.01 & 0.18 & 0.18/0.18 & 0.96/0.94 & 1.07 \\
\addlinespace
     & Biomarker & IPW & -0.54 & -0.03 & 0.16 & 0.16 & 0.94 & -- \\
     & & AIPW & -0.54 & -0.03 & 0.17 & 0.17 & 0.94 & -- \\
     & & IPW-FC & -0.54 & -0.02 & 0.16 & 0.16 & 0.95 & 1.09 \\
     & & AIPW-FC & -0.54 & -0.02 & 0.16 & 0.16/0.16 & 0.95/0.94 & 1.10 \\
\addlinespace
     & CVD & IPW & -0.27 & -0.01 & 0.20 & 0.20 & 0.94 & -- \\
     & & AIPW & -0.27 & 0.00 & 0.21 & 0.21 & 0.94 & -- \\
     & & IPW-FC & -0.27 & -0.01 & 0.19 & 0.19 & 0.94 & 1.09 \\
     & & AIPW-FC & -0.27 & -0.01 & 0.20 & 0.20/0.20 & 0.95/0.94 & 1.10 \\
\midrule
50\% & Intercept & IPW & 0.58 & 0.00 & 0.20 & 0.20 & 0.93 & -- \\
     & & AIPW & 0.58 & 0.00 & 0.21 & 0.20 & 0.92 & -- \\
     & & IPW-FC & 0.58 & 0.01 & 0.18 & 0.18 & 0.93 & 1.22 \\
     & & AIPW-FC & 0.58 & 0.00 & 0.19 & 0.19/0.18 & 0.95/0.92 & 1.24 \\
\addlinespace
     & Biomarker & IPW & -0.54 & 0.00 & 0.18 & 0.18 & 0.94 & -- \\
     & & AIPW & -0.54 & 0.00 & 0.18 & 0.18 & 0.92 & -- \\
     & & IPW-FC & -0.54 & 0.00 & 0.16 & 0.16 & 0.94 & 1.22 \\
     & & AIPW-FC & -0.54 & 0.00 & 0.16 & 0.17/0.16 & 0.94/0.92 & 1.26 \\
\addlinespace
     & CVD & IPW & -0.27 & -0.02 & 0.23 & 0.23 & 0.94 & -- \\
     & & AIPW & -0.27 & -0.02 & 0.23 & 0.24 & 0.93 & -- \\
     & & IPW-FC & -0.27 & -0.01 & 0.21 & 0.20 & 0.95 & 1.18 \\
     & & AIPW-FC & -0.27 & -0.01 & 0.21 & 0.21/0.21 & 0.95/0.93 & 1.20 \\
\midrule
65\% & Intercept & IPW & 0.58 & 0.03 & 0.22 & 0.22 & 0.93 & -- \\
     & & AIPW & 0.58 & 0.03 & 0.23 & 0.23 & 0.93 & -- \\
     & & IPW-FC & 0.58 & 0.02 & 0.19 & 0.19 & 0.95 & 1.44 \\
     & & AIPW-FC & 0.58 & 0.02 & 0.19 & 0.19/0.19 & 0.95/0.94 & 1.48 \\
\addlinespace
     & Biomarker & IPW & -0.54 & -0.01 & 0.19 & 0.19 & 0.93 & -- \\
     & & AIPW & -0.54 & -0.01 & 0.20 & 0.20 & 0.93 & -- \\
     & & IPW-FC & -0.54 & -0.01 & 0.17 & 0.17 & 0.94 & 1.34 \\
     & & AIPW-FC & -0.54 & -0.01 & 0.17 & 0.17/0.17 & 0.96/0.94 & 1.38 \\
\addlinespace
     & CVD & IPW & -0.27 & -0.02 & 0.25 & 0.25 & 0.95 & -- \\
     & & AIPW & -0.27 & -0.02 & 0.26 & 0.27 & 0.94 & -- \\
     & & IPW-FC & -0.27 & -0.02 & 0.21 & 0.21 & 0.95 & 1.43 \\
     & & AIPW-FC & -0.27 & -0.01 & 0.21 & 0.22/0.22 & 0.95/0.94 & 1.50 \\
\bottomrule
\end{tabular}
\end{table}
\end{landscape}

\section{Application to OneFlorida Data}
\label{sec:oneflorida-application}

We applied the proposed method to OneFlorida data from patients with breast cancer who had pre-existing cardiovascular disease and were classified as belonging to the high-risk frailty group following \citet{yang2025impact}. The analysis examined whether adjuvant chemotherapy and baseline frailty were associated with the chance of being favored under the prioritized comparison rule. The date of breast cancer surgery served as the index date. We adopted a standard 90-day landmark design \citep{anderson1983analysis}: patients who were alive and under observation 90 days after surgery were classified according to whether adjuvant chemotherapy had been initiated during the first 90 days, and prioritized outcomes were evaluated from day 90 onward. This defines treatment status before outcome follow-up begins and compares treated and untreated patients from a common post-surgery time origin. The analytic cohort included 1,262 patients: 231 initiated treatment within the landmark window, and 1,031 formed the control group. During three-year follow-up, 59 patients died, 662 had an inpatient hospitalization, and 723 were censored before the end of the analysis window.

The analysis used the same two-level priority rule on a 30-day grid, with death prioritized over hospitalization. Analyses were conducted over 12, 24, and 36 complete intervals, corresponding approximately to one, two, and three years, respectively. The win-regression target included two regressors: an intercept for the treated-versus-control comparison and the standardized frailty index difference. The treatment coefficient summarizes, over follow-up, the log odds that the treated patient is favored under the prioritized comparison rule in a treated-versus-control pair. Age, pre-index utilization, individual cardiovascular conditions, race and ethnicity, payer, hospitalization history, and baseline or time-updated utilization markers entered the relevant cross-fitted nuisance models for treatment assignment, censoring, future transitions, and outcome-regression augmentation.

The application used the four estimators IPW, IPW-FC, AIPW, and AIPW-FC. The future-score projection was computed by a deterministic recursion under fitted transition models over the four states defined by death status and hospitalization history. The transition, censoring, treatment, and outcome models were fitted by subject-level cross-fitting. Uncertainty was quantified with 200 subject-level nonparametric bootstrap samples, with all nuisance models and estimating equations refitted in each sample. The implementation is described in \Cref{app:oneflorida-implementation}.

\begin{table}[!htbp]
\centering
\caption{OneFlorida data analysis at the one-, two-, and three-year horizons.
The number of patients censored before each horizon is shown in parentheses.
Entries are coefficient estimates with 200 subject-bootstrap standard errors in
parentheses. Relative efficiency (RE) is the ratio of the bootstrap variance
without future-score correction to that with future-score correction for the
treated-versus-control coefficient.}
\label{tab:oneflorida-results}
\small
\setlength{\tabcolsep}{5pt}
\begin{tabular}{llccc}
\toprule
Horizon (censored)
& Estimator
& Treatment
& Frailty
& RE \\
\midrule
1 year (274) & IPW
& $0.02\ (0.21)$
& $-0.41\ (0.12)$
& -- \\
& IPW-FC
& $-0.01\ (0.19)$
& $-0.42\ (0.11)$
& $1.26$ \\
& AIPW
& $0.04\ (0.22)$
& $-0.40\ (0.13)$
& -- \\
& AIPW-FC
& $0.00\ (0.19)$
& $-0.41\ (0.11)$
& $1.35$ \\
\midrule
2 years (519) & IPW
& $0.08\ (0.21)$
& $-0.44\ (0.11)$
& -- \\
& IPW-FC
& $0.08\ (0.18)$
& $-0.48\ (0.11)$
& $1.29$ \\
& AIPW
& $0.10\ (0.21)$
& $-0.44\ (0.13)$
& -- \\
& AIPW-FC
& $0.10\ (0.19)$
& $-0.47\ (0.12)$
& $1.31$ \\
\midrule
3 years (723) & IPW
& $0.10\ (0.21)$
& $-0.45\ (0.12)$
& -- \\
& IPW-FC
& $0.12\ (0.17)$
& $-0.49\ (0.10)$
& $1.47$ \\
& AIPW
& $0.13\ (0.21)$
& $-0.44\ (0.13)$
& -- \\
& AIPW-FC
& $0.13\ (0.18)$
& $-0.48\ (0.12)$
& $1.48$ \\
\bottomrule
\end{tabular}
\end{table}

Treatment coefficient estimates were similar across the four estimators, and all bootstrap intervals included zero. At the three-year horizon, the AIPW-FC coefficient was 0.13, corresponding to an estimated odds ratio of 1.14 for being favored under the prioritized comparison rule ($95\%$ bootstrap confidence interval: 0.84, 1.71), with no clear evidence of a treatment-associated difference. Higher frailty was associated with a lower chance of being favored under the prioritized comparison rule: the AIPW-FC frailty coefficient was $-0.48$ ($95\%$ confidence interval: $-0.69$, $-0.22$).

The main difference among the estimators was precision. For the treatment coefficient, the relative efficiencies of IPW-FC versus IPW were 1.26, 1.29, and 1.47 at one, two, and three years, respectively; the corresponding relative efficiencies for AIPW-FC versus AIPW were 1.35, 1.31, and 1.48. At three years, future-score correction reduced the bootstrap standard error from 0.21 to 0.17 for IPW and from 0.21 to 0.18 for AIPW. The similar point estimates and smaller standard errors indicate that the gain came from recovering future pairwise score contributions after censoring while preserving the target comparison. Additional analyses under reduced specifications of the nuisance models yielded similar conclusions and are reported in \Cref{app:oneflorida-sensitivity}.

\section{Discussion}

Future-score correction recovers information that inverse-probability-of-censoring weighting leaves unused. When the censoring model is correctly specified, inverse-probability-of-censoring weighting recovers the complete-data score in expectation, but information in the pre-censoring history about subsequent score contributions remains unused. Future-score correction recovers this information through the conditional expectation of the remaining complete-data score given the observed pair history. This mechanism suggests larger efficiency gains when censoring removes a substantial portion of the future score and the pre-censoring history remains informative about what follows. The simulations showed this pattern, with efficiency gains increasing as censoring became more severe while bias remained small. In the OneFlorida analysis, future-score correction similarly yielded smaller standard errors across follow-up horizons while producing comparable coefficient estimates. We further show that, under the stated rate conditions, the AIPW-FC estimator has the efficient influence function for $\beta_0$ and attains the semiparametric efficiency bound, so no regular asymptotically linear estimator for the same target based on the same observed information can have smaller asymptotic variance.

The target $\beta_0$ is a complete-data quantity defined independently of the censoring process. Censoring therefore determines which score contributions are observed and which must be recovered, but it does not determine the target itself. The proposed estimating equation preserves this complete-data target while providing two distinct double-robustness properties: treatment adjustment is valid when either the propensity score or baseline outcome regression is correct, and censoring adjustment is valid when either the censoring model or future-score projection is correct.

Recovering the information carried by the pre-censoring history requires estimating the future-score projection. Its quality therefore affects how much of the available efficiency gain is realized. In the death-hospitalization setting considered here, the projection can be evaluated by finite-state recursion; with more complex histories, it can instead be estimated directly or constructed from models for future trajectories. Future work can develop scalable implementations for such settings, extend the framework to allow censoring to occur separately for different outcome levels, and allow the regression association itself to vary over follow-up.

\appendix

\section{Properties of the Complete-Data Estimand}
\label{app:complete-data-estimand-properties}

\begin{proposition}[Target Uniqueness]
\label{prop:target-unique}
Let the population score be the map $\beta\mapsto \E\{\omega^{1,0}(e_0)S^{\circ}(\beta)\}$. Suppose it is finite and admits dominated differentiation on compact line segments in $\beta$. Suppose also that, for every $\beta\in\R^{p+1}$ and every nonzero $\rho\in\R^{p+1}$,
\begin{equation}
\label{eq:target-unique-cond}
\E\!\left[
\omega^{1,0}(e_0)
\int_0^\tau
R^{\circ}(t)\,
\expit\!\bigl\{\beta^\top\widetilde X\bigr\}
\left\{
1-\expit\!\bigl\{\beta^\top\widetilde X\bigr\}
\right\}
\bigl(\rho^\top\widetilde X\bigr)^2
\;dt
\right]
>0.
\end{equation}
Then the population score map $\beta\mapsto \E\{\omega^{1,0}(e_0)S^{\circ}(\beta)\}$ is strictly monotone, so \cref{eq:target} has at most one solution in $\R^{p+1}$.
\end{proposition}

\begin{proof}
Differentiating \cref{eq:full-score} under the expectation yields
\[
\frac{\partial}{\partial\beta^\top}
\E\!\left\{\omega^{1,0}(e_0)S^{\circ}(\beta)\right\}
=
-\E\!\left[
\omega^{1,0}(e_0)
\int_0^\tau
R^{\circ}(t)\,
\expit\!\bigl\{\beta^\top\widetilde X\bigr\}
\left\{
1-\expit\!\bigl\{\beta^\top\widetilde X\bigr\}
\right\}
\widetilde X\widetilde X^\top
\;dt
\right].
\]
Hence for every nonzero $\rho\in\R^{p+1}$,
\begin{align*}
&\rho^\top
\frac{\partial}{\partial\beta^\top}
\E\!\left\{\omega^{1,0}(e_0)S^{\circ}(\beta)\right\}
\rho
\\
&\quad=
-\E\!\left[
\omega^{1,0}(e_0)
\int_0^\tau
R^{\circ}(t)\,
\expit\!\bigl\{\beta^\top\widetilde X\bigr\}
\left\{
1-\expit\!\bigl\{\beta^\top\widetilde X\bigr\}
\right\}
\bigl(\rho^\top\widetilde X\bigr)^2
\;dt
\right]
<0.
\end{align*}
By \cref{eq:target-unique-cond}, the Jacobian of the population score map is negative definite at every $\beta$, so the map is strictly monotone. For $\beta_1\neq\beta_2$, set $d=\beta_1-\beta_2$. The mean-value theorem applied to
\[
t\longmapsto
d^\top\E\!\left\{\omega^{1,0}(e_0)S^{\circ}(\beta_2+td)\right\}
\]
yields $c\in(0,1)$ and $\bar\beta=\beta_2+cd$ satisfying
\[
\begin{aligned}
&d^\top\left[
\E\!\left\{\omega^{1,0}(e_0)S^{\circ}(\beta_1)\right\}
-
\E\!\left\{\omega^{1,0}(e_0)S^{\circ}(\beta_2)\right\}
\right]
\\
&\qquad=
d^\top
\left.
\frac{\partial}{\partial\beta^\top}
\E\!\left\{\omega^{1,0}(e_0)S^{\circ}(\beta)\right\}
\right|_{\beta=\bar\beta}
d
<0.
\end{aligned}
\]
Consequently, equality of the population scores at $\beta_1$ and $\beta_2$ implies $\beta_1=\beta_2$, and \cref{eq:target} has at most one solution.
\end{proof}

\begin{proposition}[Proportional Win-Fraction Target]
\label{prop:pwf-align}
Suppose there exists $\beta^{\ast}\in\R^{p+1}$ such that the complete-data residual integrand in \cref{eq:full-score}, evaluated at $\beta^{\ast}$, is integrable over $[0,\tau]$. Suppose further that, for every $t\in[0,\tau]$,
\begin{equation}
\label{eq:pwf-align-cond}
\E\!\left[
\omega^{1,0}(e_0)\widetilde X
\left\{
W^{\circ}(t)
-
R^{\circ}(t)\expit\!\bigl\{(\beta^{\ast})^\top\widetilde X\bigr\}
\right\}
\right]
=
0.
\end{equation}
Then
\[
\E\!\left\{\omega^{1,0}(e_0)S^{\circ}(\beta^{\ast})\right\}=0.
\]
If, in addition, \cref{eq:target-unique-cond} holds, then $\beta_0=\beta^{\ast}$.
\end{proposition}

\begin{proof}
By \cref{eq:pwf-align-cond}, the integrand in \cref{eq:full-score} has mean zero at $\beta=\beta^{\ast}$ for every $t\in[0,\tau]$. Integrating that identity over $t$ and applying Fubini's theorem yields
\[
\E\!\left[
\omega^{1,0}(e_0)\widetilde X
\int_0^\tau
\left\{
W^{\circ}(t)
-
R^{\circ}(t)\expit\!\bigl\{(\beta^{\ast})^\top\widetilde X\bigr\}
\right\}\,dt
\right]
=
0.
\]
By \cref{eq:full-score}, this is $\E\{\omega^{1,0}(e_0)S^{\circ}(\beta^{\ast})\}=0$; hence $\beta^{\ast}$ solves the population score equation. If \cref{eq:target-unique-cond} also holds, \cref{prop:target-unique} implies uniqueness, and therefore $\beta_0=\beta^{\ast}$.
\end{proof}

\section{Large-Sample Theory and Proofs}
\label{app:large-sample-theory}

\subsection{Assumptions}

\begin{assumption}[Consistency and No Interference]
\label{ass:consistency}
For each subject, the observed outcome process equals the potential outcome process under the realized treatment, and one subject's outcome process is not affected by other subjects' treatments.
\end{assumption}

\begin{assumption}[Independent Subjects]
\label{ass:subject-independence}
Subjects are sampled independently from a common population.
\end{assumption}

\begin{assumption}[Sequentially Independent Censoring]
\label{ass:censoring}
For each subject and interval, conditional on $\mathcal H_{i,\ell}^{+}$, the censoring jump at $t_\ell$ is independent of the future complete-data outcome path.
\end{assumption}

\begin{assumption}[Censoring Positivity]
\label{ass:censorpositivity}
There exists $\epsilon>0$ such that every true or working censoring survival for a pair used in the estimating equation is at least $\epsilon$ on the support of the observed pair histories.
\end{assumption}

\begin{assumption}[Finite Second Moments]
\label{ass:second-moment}
The baseline covariates satisfy \( \E\!\left\{\|X\|^2\right\}<\infty. \)
\end{assumption}

\begin{assumption}[Treatment Exchangeability and Positivity]
\label{ass:treatment}
Conditional on $X$, treatment assignment is independent of the potential outcome and censoring processes, and there exists $\epsilon>0$ such that
\[
  \epsilon \le e_0(X)\le 1-\epsilon
  \qquad\text{almost surely.}
\]
\end{assumption}

\begin{assumption}[Pathwise Regularity]
\label{ass:pathwise-regularity}
The observed-data functional $\psi_\beta(P)$ is pathwise differentiable at $P_0$ along regular parametric submodels, uniformly for $\beta$ in a neighborhood of $\beta_0$. For each such submodel $P_\varepsilon$, the map
\[
  (\varepsilon,\eta)
  \longmapsto
  \E_{P_\varepsilon\otimes P_\varepsilon}
  \{K(\beta;\eta)\},
\]
where $K$ is the ordered-pair kernel defined in \eqref{eq:aipw-kernel}, satisfies the chain rule along $\eta=\eta_{P_\varepsilon}^{\mathrm{obs}}(\beta)$ and permits interchange of differentiation and integration, with square-integrable derivatives.
\end{assumption}

\begin{assumption}[Parameter Space and Kernel Regularity]
\label{ass:kernel-regularity}
The estimating equation is solved over a compact, convex parameter set $\Theta\subset\R^{p+1}$, with $\beta_0$ in its interior. The set $\Theta$ represents a sufficiently broad and bounded search region for $\beta$. The kernel satisfies $\E\{\sup_{\beta\in\Theta}\|K(\beta;\eta_0)\|^2\}<\infty$, and the maps $\beta\mapsto\psi_\beta(P_0)$ and $\beta\mapsto \E\{K(\beta;\eta_0)\}$ are continuously differentiable on a neighborhood of $\Theta$. The kernel is locally Lipschitz in the nuisance functions, with a square-integrable envelope, in a neighborhood of their population values. The path
$\beta\mapsto K\{\beta;\eta_{P_0}^{\mathrm{obs}}(\beta)\}$ is Lipschitz on $\Theta$
with a square-integrable Lipschitz envelope under $P_0\otimes P_0$. In addition,
\[
  \left.
  \frac{\partial}{\partial\beta^\top}
  \psi_\beta(P_0)
  \right|_{\beta=\beta_0}
\]
is nonsingular.
\end{assumption}

\begin{assumption}[Cross-Fitting and Derivative Convergence]
\label{ass:crossfit-regularity}
The subjects are partitioned into $F$ folds, where $F\geq 3$ is fixed and each fold has size proportional to $n$. For each pair of folds $(f,g)$, let $\widehat\eta^{-(f,g)}$ denote the corresponding nuisance fit, which is based only on subjects outside those folds and is independent of the records in the evaluation folds. Conditional on its training sample, the kernel path $\beta\mapsto K\{\beta;\widehat\eta^{-(f,g)}(\beta)\}$ is Lipschitz on $\Theta$ with an $L_2(P_0\otimes P_0)$ Lipschitz envelope whose norm is uniformly bounded in probability. After construction of the cross-fitted nuisance fits, there exists a fixed neighborhood $\mathcal N\subset\Theta$ of $\beta_0$ such that
\[
  \sup_{\beta\in\mathcal N}
  \left\|
    \frac{\partial}{\partial\beta^\top}\widehat U_n(\beta)
    -
    \frac{\partial}{\partial\beta^\top}\psi_\beta(P_0)
  \right\|
  =
  o_p(1).
\]
\end{assumption}

\subsection{Censoring Double Robustness}

Under Assumption~\ref{ass:second-moment}, $\Psi_\ell(\beta)$ is square-integrable, so $\Phi_{0,\ell}(\beta)$ is well defined.

\begin{proof}[Proof of \headingcref{thm:future-score-dr}{Theorem}.]
Under the correctly specified censoring model $G=G_0$, sequentially independent censoring renders $dM_{\ell}^C(G_0)$ a martingale increment with respect to the pair history. Because $\Phi_\ell(\beta)$ is predictable, the multiplier $Y_{\ell-1}G_{0,\ell}^{-1}\Phi_\ell(\beta)$ is $\mathcal H_\ell^{+}$-measurable. Hence, by iterated expectation,
\[
\begin{aligned}
\E\!\left\{
  \frac{Y_{\ell-1}}{G_{0,\ell}}
  \Phi_\ell(\beta)\,dM_{\ell}^{C}(G_0)
\right\}
&=
\E\!\left[
  \frac{Y_{\ell-1}}{G_{0,\ell}}
  \Phi_\ell(\beta)\,
  \E\{dM_{\ell}^{C}(G_0)\mid\mathcal H_\ell^{+}\}
\right] \\
&=0.
\end{aligned}
\]
The usual IPCW identity also implies
\[
  \E\!\left\{
  \sum_{\ell=1}^{M}\frac{Y_{\ell-1}}{G_{0,\ell-1}}S_\ell^\circ(\beta)
  \right\}
  =
  \E\!\left\{\sum_{\ell=1}^{M}S_\ell^\circ(\beta)\right\}.
\]
Therefore,
\[
\begin{aligned}
  \E\{S^{\mathrm{FC}}(\beta;G_0,\Phi)\}
  &=
  \E\!\left\{
    \sum_{\ell=1}^{M}
    \frac{Y_{\ell-1}}{G_{0,\ell-1}}S_\ell^\circ(\beta)
  \right\}
  +
  \sum_{\ell=1}^{M-1}
  \E\!\left\{
    \frac{Y_{\ell-1}}{G_{0,\ell}}
    \Phi_\ell(\beta)\,dM_\ell^C(G_0)
  \right\}
  \\
  &=
  \E\!\left\{\sum_{\ell=1}^{M}S_\ell^\circ(\beta)\right\}
  =\E\{S^{\circ}(\beta)\}.
\end{aligned}
\]

Under a correct predictable projection $\Phi=\Phi_0$, let $G$ be any sequence of censoring survival probabilities bounded away from zero, with interval values $G_\ell$. On the at-risk set $Y_{\ell-1}=1$,
\[
  dM_\ell^C(G)
  =
  \Delta N_\ell^C-\lambda_\ell^C(G).
\]
By definition, $\Phi_{0,\ell}(\beta)=\E\{\Psi_\ell(\beta)\mid\mathcal H_\ell^{+}\}$. Under sequentially independent censoring, $1-\Delta N_\ell^C$ is independent of $\Psi_\ell(\beta)-\Phi_{0,\ell}(\beta)$ given $\mathcal H_\ell^{+}$, while
\[
\begin{aligned}
  \frac{1-\Delta N_\ell^C}{1-\lambda_\ell^C(G)}\Psi_\ell(\beta)
  +
  \frac{\Delta N_\ell^C-\lambda_\ell^C(G)}
       {1-\lambda_\ell^C(G)}\Phi_{0,\ell}(\beta)
  &=
  \Phi_{0,\ell}(\beta)
  +
  \frac{1-\Delta N_\ell^C}{1-\lambda_\ell^C(G)}
  \{\Psi_\ell(\beta)-\Phi_{0,\ell}(\beta)\}.
\end{aligned}
\]
Therefore,
\[
\begin{aligned}
  &\E\!\left[
    \frac{1-\Delta N_\ell^C}{1-\lambda_\ell^C(G)}\Psi_\ell(\beta)
    +
    \frac{\Delta N_\ell^C-\lambda_\ell^C(G)}
         {1-\lambda_\ell^C(G)}\Phi_{0,\ell}(\beta)
    \,\middle|\,
    \mathcal H_\ell^{+}
  \right]
  \\
  &\quad=
  \Phi_{0,\ell}(\beta)
  +
  \E\!\left[
    \frac{1-\Delta N_\ell^C}{1-\lambda_\ell^C(G)}
    \{\Psi_\ell(\beta)-\Phi_{0,\ell}(\beta)\}
    \,\middle|\,
    \mathcal H_\ell^{+}
  \right]
  \\
  &\quad=
  \Phi_{0,\ell}(\beta)
  +
  \frac{
    \E(1-\Delta N_\ell^C\mid\mathcal H_\ell^{+})
  }{
    1-\lambda_\ell^C(G)
  }
  \E\!\left[
    \Psi_\ell(\beta)-\Phi_{0,\ell}(\beta)
    \mid\mathcal H_\ell^{+}
  \right]
  \\
  &\quad=
  \Phi_{0,\ell}(\beta).
\end{aligned}
\]
Multiplying this identity by $Y_{\ell-1}/G_{\ell-1}$ and using the interval updates $Y_{\ell-1}(1-\Delta N_\ell^C)=Y_\ell$ and $G_{\ell-1}\{1-\lambda_\ell^C(G)\}=G_\ell$ yields
\[
  \E\!\left[
  \frac{Y_{\ell}}{G_{\ell}}\Psi_\ell(\beta)
  +
  \frac{Y_{\ell-1}}{G_{\ell}}
  \Phi_{0,\ell}(\beta)\,dM_\ell^C(G)
  \middle|\mathcal H_{\ell}^{+}
  \right]
  =
  \frac{Y_{\ell-1}}{G_{\ell-1}}\Phi_{0,\ell}(\beta).
\]

Taking expectations of both sides in the preceding display gives
\[
\begin{aligned}
  &\E\!\left[
    \frac{Y_{\ell}}{G_{\ell}}\Psi_\ell(\beta)
    +
    \frac{Y_{\ell-1}}{G_{\ell}}
    \Phi_{0,\ell}(\beta)\,dM_\ell^C(G)
  \right]
  \\
  &\quad=
  \E\!\left[
    \E\!\left\{
      \frac{Y_{\ell}}{G_{\ell}}\Psi_\ell(\beta)
      +
      \frac{Y_{\ell-1}}{G_{\ell}}
      \Phi_{0,\ell}(\beta)\,dM_\ell^C(G)
      \,\middle|\,
      \mathcal H_\ell^{+}
    \right\}
  \right]
  \\
  &\quad=
  \E\!\left\{
    \frac{Y_{\ell-1}}{G_{\ell-1}}\Phi_{0,\ell}(\beta)
  \right\}
  =
  \E\!\left\{
    \frac{Y_{\ell-1}}{G_{\ell-1}}\Psi_\ell(\beta)
  \right\}.
\end{aligned}
\]
Together with $\Psi_{\ell-1}(\beta)=S_\ell^\circ(\beta)+\Psi_\ell(\beta)$, this gives
\[
\begin{aligned}
  \E\!\left[
    \frac{Y_{\ell-1}}{G_{\ell-1}}\Psi_{\ell-1}(\beta)
  \right]
  =
  \E\!\left[
    \frac{Y_{\ell-1}}{G_{\ell-1}}S_\ell^\circ(\beta)
    +
    \frac{Y_{\ell}}{G_{\ell}}\Psi_\ell(\beta)
    +
    \frac{Y_{\ell-1}}{G_{\ell}}
    \Phi_{0,\ell}(\beta)\,dM_\ell^C(G)
  \right].
\end{aligned}
\]
Starting from $\Psi_{M-1}(\beta)=S_M^\circ(\beta)$ and using $Y_\ell=G_\ell=1$ at $\ell=0$, applying this recursion successively for $\ell=M-1,\ldots,1$ yields
\[
  \E\!\left[
  \sum_{\ell=1}^{M}
  \frac{Y_{\ell-1}}{G_{\ell-1}}S_\ell^\circ(\beta)
  +
  \sum_{\ell=1}^{M-1}
  \frac{Y_{\ell-1}}{G_{\ell}}
  \Phi_{0,\ell}(\beta)\,dM_\ell^C(G)
  \right]
  =
  \E\!\left\{\sum_{\ell=1}^{M}S_\ell^\circ(\beta)\right\}.
\]
Equivalently,
\[
  \E\{S^{\mathrm{FC}}(\beta;G,\Phi_0)\}
  =
  \E\{S^{\circ}(\beta)\},
\]
which proves the second assertion.
\end{proof}

The same double-robustness property holds conditional on $(X_1,X_2)$, since each of the preceding arguments can be carried out conditional on $(X_1,X_2)$. For treated-versus-control pairs, the corrected score has conditional mean $m_0(X_1,X_2;\beta)$ given $(X_1,X_2)$ whenever either $G=G_0$ or $\Phi=\Phi_0$.

\subsection{Double Robustness}

\begin{proof}[Proof of \headingcref{thm:aipw-dr}{Theorem}.]
For $x=(x_1,x_0)$, write
\[
  r_e(x)
  =
  \frac{e_0(x_1)\{1-e_0(x_0)\}}{e(x_1)\{1-e(x_0)\}},
\]
and
\[
  \mu_{G,\Phi}(x)
  =
  \E\{S^{\mathrm{FC}}(\beta_0;G,\Phi)
  \mid X_1=x_1,X_2=x_0,A_1=1,A_2=0\}.
\]
Conditioning first on $(X_1,X_2)$ and then on $(A_1,A_2)$ gives the conditional representation
\[
  \E\{K(\beta_0;\eta)\mid X_1=x_1,X_2=x_0\}
  =
  r_e(x)\{\mu_{G,\Phi}(x)-m(x)\}+m(x),
\]
where $m(x)=m(x_1,x_0;\beta_0)$.

By the censoring double robustness result in \Cref{thm:future-score-dr}, $\mu_{G,\Phi}(x)=m_0(x)$ whenever either $G=G_0$ or $\Phi=\Phi_0$. If, in addition, $e=e_0$, then $r_e(x)=1$ and the conditional mean equals $m_0(x)$. When $m=m_0$, the difference $\mu_{G,\Phi}(x)-m(x)$ vanishes, so the conditional mean again equals $m_0(x)$. Integrating over the target covariate distribution yields
\[
  \E\{K(\beta_0;\eta)\}
  =
  \E\{m_0(X_1,X_2)\}
  =
  \E\{\omega^{1,0}(e_0)S^{\circ}(\beta_0)\}
  =
  0,
\]
where the final equality is the complete-data target equation.

Neyman orthogonality \citep{chernozhukov2018double} follows from the conditional representation of \(\E\{K(\beta_0;\eta)\mid X_1=x_1,X_2=x_0\}\). Consider regular one-dimensional paths that vary one nuisance component at a time and leave the remaining components fixed at their true values.
\[
  \E\{K(\beta_0;e,G,\Phi,m)\mid X_1=x_1,X_2=x_0\}
  =
  r_e(x_1,x_0)\{\mu_{G,\Phi}(x_1,x_0)-m(x_1,x_0)\}+m(x_1,x_0).
\]
This conditional mean equals \(m_0(x_1,x_0)\) along each coordinate path. If only the propensity score varies, \(\mu_{G_0,\Phi_0}=m_0\); if only \(m\) varies, then \(r_{e_0}=1\) and \(\mu_{G_0,\Phi_0}=m_0\). If only \(G\) or only \(\Phi\) varies, \Cref{thm:future-score-dr} implies \(\mu_{G_t,\Phi_0}=m_0\) or \(\mu_{G_0,\Phi_t}=m_0\), respectively. Thus the conditional mean is constant in the path parameter for each coordinate direction,
\[
  \left.
  \frac{d}{dt}\E\{K(\beta_0;\eta_t)\}
  \right|_{t=0}
  =0.
\]
Linearity of the Gateaux derivative then extends the conclusion to every regular nuisance direction at \(\eta_0\).

It remains to establish the second-order drift bound. For general $\eta$, the drift decomposes as
\[
\begin{aligned}
  &\E\{K(\beta_0;\eta)-K(\beta_0;\eta_0)\}
  \\
  &\quad=
  \E\!\left[
    \E\{K(\beta_0;\eta)-K(\beta_0;\eta_0)
      \mid X_1,X_2\}
  \right]
  \\
  &\quad=
  \E\!\left[
    r_e(X_1,X_2)
    \{\mu_{G,\Phi}(X_1,X_2)-m(X_1,X_2)\}
    +m(X_1,X_2)-m_0(X_1,X_2)
  \right]
  \\
  &\quad=
  \E\!\left[
    r_e(X_1,X_2)
    \{\mu_{G,\Phi}(X_1,X_2)-m_0(X_1,X_2)\}
  \right]
  \\
  &\qquad+
  \E\!\left[
    \{r_e(X_1,X_2)-1\}
    \{m_0(X_1,X_2)-m(X_1,X_2)\}
  \right].
\end{aligned}
\]
As $e$ and $e_0$ are bounded away from zero and one, $\|r_e-1\|_2 = \left\| \frac{ e_0(X_1)\{1-e_0(X_2)\} -e(X_1)\{1-e(X_2)\}} {e(X_1)\{1-e(X_2)\}} \right\|_2 \lesssim \left\| |e(X_1)-e_0(X_1)| + |e(X_2)-e_0(X_2)| \right\|_2 \le \|e(X_1)-e_0(X_1)\|_2 + \|e(X_2)-e_0(X_2)\|_2 = 2\|e-e_0\|_2 = O(\|e-e_0\|_2)$, so Cauchy--Schwarz bounds the second term by $O\!\left(\|e-e_0\|\,\|m-m_0\|\right)$.

For the first term, $\mu_{G,\Phi_0}=m_0$ gives
\[
  \mu_{G,\Phi}-m_0
  =
  \mu_{G,\Phi}-\mu_{G,\Phi_0}.
\]
The IPCW terms in $S^{\mathrm{FC}}(\beta_0;G,\Phi)$ and $S^{\mathrm{FC}}(\beta_0;G,\Phi_0)$ are identical, so subtracting the two scores leaves only the difference between their augmentation terms. Therefore, componentwise,
\[
  \mu_{G,\Phi}(X_1,X_2)-m_0(X_1,X_2)
  =
  \E\!\left[
  \sum_{\ell=1}^{M-1}
  \frac{Y_{\ell-1}}{G_\ell}
  \{\Phi_\ell-\Phi_{0,\ell}\}
  dM_\ell^C(G)
  \middle|X_1,X_2,A_1=1,A_2=0
  \right].
\]
Let $\lambda_\ell^C(G)$ and $\lambda_\ell^C$ denote the working and true hazards for pairwise censoring in interval $\ell$. Under the true censoring model,
\[
  \E\{dM_\ell^C(G)\mid \mathcal H_\ell^{+}\}
  =
  Y_{\ell-1}\{\lambda_\ell^C-\lambda_\ell^C(G)\}.
\]

Since $Y_{\ell-1}G_\ell^{-1}\{\Phi_\ell-\Phi_{0,\ell}\}$ is $\mathcal H_\ell^{+}$-measurable, conditioning first on $\mathcal H_\ell^{+}$ gives
\[
\begin{aligned}
  \mu_{G,\Phi}(X_1,X_2)-m_0(X_1,X_2)
  &=
  \E\!\left[
  \sum_{\ell=1}^{M-1}
  \frac{Y_{\ell-1}}{G_\ell}
  \{\Phi_\ell-\Phi_{0,\ell}\}
  \E\{dM_\ell^C(G)\mid\mathcal H_\ell^{+}\}
  \middle|X_1,X_2,A_1=1,A_2=0
  \right] \\
  &=
  \E\!\left[
  \sum_{\ell=1}^{M-1}
  \frac{Y_{\ell-1}^2}{G_\ell}
  \{\lambda_\ell^C-\lambda_\ell^C(G)\}
  \{\Phi_\ell-\Phi_{0,\ell}\}
  \middle|X_1,X_2,A_1=1,A_2=0
  \right] \\
  &=
  \E\!\left[
  \sum_{\ell=1}^{M-1}
  \frac{Y_{\ell-1}}{G_\ell}
  \{\lambda_\ell^C-\lambda_\ell^C(G)\}
  \{\Phi_\ell-\Phi_{0,\ell}\}
  \middle|X_1,X_2,A_1=1,A_2=0
  \right].
\end{aligned}
\]
Since
\[
  \lambda_\ell^C(G)
  =
  1-\frac{G_\ell}{G_{\ell-1}},
  \qquad
  \lambda_\ell^C
  =
  1-\frac{G_{0,\ell}}{G_{0,\ell-1}},
\]
we have
\[
\begin{aligned}
  \lambda_\ell^C-\lambda_\ell^C(G)
  &=
  \frac{G_\ell}{G_{\ell-1}}
  -
  \frac{G_{0,\ell}}{G_{0,\ell-1}} \\
  &=
  \frac{G_\ell-G_{0,\ell}}{G_{\ell-1}}
  +
  \frac{
    G_{0,\ell}\{G_{0,\ell-1}-G_{\ell-1}\}
  }{
    G_{\ell-1}G_{0,\ell-1}
  }.
\end{aligned}
\]
Positivity and the preceding expansion give
\[
\begin{aligned}
  \left\|
  \frac{Y_{\ell-1}}{G_\ell}
  \{\lambda_\ell^C-\lambda_\ell^C(G)\}
  \right\|_2
  &=
  O\!\left(
  \|G_\ell-G_{0,\ell}\|_2
  +
  \|G_{\ell-1}-G_{0,\ell-1}\|_2
  \right) \\
  &=
  O\!\left(\|G-G_0\|_2\right).
\end{aligned}
\]

By Cauchy--Schwarz,
\[
\E\!\left[
\left|
\frac{Y_{\ell-1}}{G_\ell}
\{\lambda_\ell^C-\lambda_\ell^C(G)\}
\{\Phi_\ell-\Phi_{0,\ell}\}
\right|
\right]
\le
\left\|
\frac{Y_{\ell-1}}{G_\ell}
\{\lambda_\ell^C-\lambda_\ell^C(G)\}
\right\|_2
\|\Phi_\ell-\Phi_{0,\ell}\|_2.
\]
Hence,
\[
\begin{aligned}
&\left|
\E\!\left[
r_e(X_1,X_2)
\{\mu_{G,\Phi}(X_1,X_2)-m_0(X_1,X_2)\}
\right]
\right|
\\
&\quad=
\left|
\E\!\left[
r_e(X_1,X_2)
\E\!\left\{
\sum_{\ell=1}^{M-1}
\frac{Y_{\ell-1}}{G_\ell}
\{\lambda_\ell^C-\lambda_\ell^C(G)\}
\{\Phi_\ell-\Phi_{0,\ell}\}
\middle|
X_1,X_2,A_1=1,A_2=0
\right\}
\right]
\right|
\\
&\quad=
O\!\left(
\sum_{\ell=1}^{M-1}
\left\|
\frac{Y_{\ell-1}}{G_\ell}
\{\lambda_\ell^C-\lambda_\ell^C(G)\}
\right\|_2
\|\Phi_\ell-\Phi_{0,\ell}\|_2
\right)
\\
&\quad=
O\!\left(
\|G-G_0\|_2\,\|\Phi-\Phi_0\|_2
\right).
\end{aligned}
\]

Combining this bound with the bound for the treatment component yields
\[
  \E\{K(\beta_0;\eta)-K(\beta_0;\eta_0)\}
  =
  O\!\left(
  \|e-e_0\|\,\|m-m_0\|
  +
  \|G-G_0\|\,\|\Phi-\Phi_0\|
  \right),
\]
as claimed.
\end{proof}

\subsection{Efficient Influence Function}

We first verify \eqref{eq:truth-law-nuisance}. For $r>\ell$, sequentially independent censoring gives
\[
\E\!\left[
  \left.
  \frac{Y_{r-1}G_{0,\ell-1}}{G_{0,r-1}}
  S_r^\circ(\beta_0)
  \,\right|\,
  \mathcal H_\ell^+
\right]
=
\E\{S_r^\circ(\beta_0)\mid\mathcal H_\ell^+\}.
\]
Hence
\[
\begin{aligned}
\Phi_{P_0,\ell}^{\mathrm{obs}}(\beta_0)
&=
\E\!\left\{
\sum_{r>\ell}S_r^\circ(\beta_0)
\mid\mathcal H_\ell^+
\right\} \\
&=
\E\{\Psi_\ell(\beta_0)\mid\mathcal H_\ell^+\}
=
\Phi_{0,\ell}(\beta_0).
\end{aligned}
\]
Also, by the definition of $m_{P_0}^{\mathrm{obs}}$ and $G_{P_0}=G_0$,
\[
m_{P_0}^{\mathrm{obs}}(x_1,x_0;\beta_0)
=
\E\!\left[
S^{\mathrm{FC}}\{\beta_0;G_0,\Phi_{P_0}^{\mathrm{obs}}(\beta_0)\}
\mid X_1=x_1,X_2=x_0,A_1=1,A_2=0
\right].
\]
Carrying the proof of \eqref{eq:censoring-correct-G} conditionally gives
\[
m_{P_0}^{\mathrm{obs}}(x_1,x_0;\beta_0)
=
m_0(x_1,x_0;\beta_0).
\]

\begin{proof}[Proof of \headingcref{prop:canonical-gradient}{Proposition}.]
Fix $\beta=\beta_0$. Let $\{P_\varepsilon:\varepsilon\in(-\delta,\delta)\}$ be an arbitrary regular parametric submodel through $P_0$, with score $s(O)\in L_2^0(P_0)$, and write $\eta_\varepsilon=\eta_{P_\varepsilon}^{\mathrm{obs}}(\beta_0)$. By \eqref{eq:truth-law-nuisance}, together with $e_{P_0}=e_0$ and $G_{P_0}=G_0$, this path passes through $\eta_0$ at $\varepsilon=0$. Define
\[
  Q(\varepsilon,\eta)
  =
  \E_{P_\varepsilon\otimes P_\varepsilon}
  \{K(\beta_0;\eta)\}.
\]
By \eqref{eq:observed-moment-functional}, $\psi_{\beta_0}(P_\varepsilon)=Q(\varepsilon,\eta_\varepsilon)$. Assumption~\ref{ass:pathwise-regularity} permits the pathwise chain rule:
\begin{equation}
\label{eq:pathwise-chain-rule}
  \left.
  \frac{d}{d\varepsilon}
  \psi_{\beta_0}(P_\varepsilon)
  \right|_{\varepsilon=0}
  =
  \left.
  \frac{\partial}{\partial\varepsilon}
  Q(\varepsilon,\eta_0)
  \right|_{\varepsilon=0}
  +
  \partial_\eta Q(0,\eta_0)[\dot\eta],
\end{equation}
where $\dot\eta$ is the derivative of $\eta_\varepsilon$ at zero. In the first term, $\eta$ is fixed at $\eta_0$ and only $P_\varepsilon\otimes P_\varepsilon$ varies; the second term is the derivative along the nuisance path $\eta_\varepsilon$.

The second term in \eqref{eq:pathwise-chain-rule} is zero by the orthogonality calculation underlying Theorem~\ref{thm:aipw-dr}. Differentiating the product-remainder expansion of the augmented moment at $\eta_0$ eliminates all first-order nuisance perturbations, and hence
\[
  \partial_\eta Q(0,\eta_0)[\dot\eta]=0.
\]

Since the score of $P_\varepsilon\otimes P_\varepsilon$ is $s(O_1)+s(O_2)$, the remaining derivative is
\begin{align*}
  \left.
  \frac{\partial}{\partial\varepsilon}
  Q(\varepsilon,\eta_0)
  \right|_{\varepsilon=0}
  &=
  \E\!\left[
  K_{12}(\beta_0;\eta_0)\{s(O_1)+s(O_2)\}
  \right] \\
  &=
  \E\!\left[
  \{K_{12}(\beta_0;\eta_0)+K_{21}(\beta_0;\eta_0)\}
  s(O_1)
  \right] \\
  &=
  \E\!\left[
  2\E\{\bar K(O_1,O_2)\mid O_1\}
  s(O_1)
  \right].
\end{align*}
The second equality follows by exchanging the labels of the two independent subjects in the term containing $s(O_2)$. The third equality follows by substituting
$K_{12}(\beta_0;\eta_0)+K_{21}(\beta_0;\eta_0)=2\bar K(O_1,O_2)$
and applying iterated expectation conditional on $O_1$. The definition of $\kappa$ then gives
\[
  \left.
  \frac{d}{d\varepsilon}
  \psi_{\beta_0}(P_\varepsilon)
  \right|_{\varepsilon=0}
  =
  \E\{D_\psi^{\mathrm{eff}}(O)s(O)\},
  \qquad
  D_\psi^{\mathrm{eff}}(O)=2\kappa(O).
\]

Jensen's inequality and the kernel second-moment condition imply that $D_\psi^{\mathrm{eff}}\in L_2^0(P_0)$. For the nonparametric observed-data model considered here, the closure of the observed-data tangent space at $P_0$ is $L_2^0(P_0)$. Since $D_\psi^{\mathrm{eff}}$ represents the pathwise derivative above, it is the canonical gradient of the moment.

We next derive the canonical gradient of the root functional $\beta(P)$. The pathwise chain rule in the $\beta$ direction yields
\[
  \begin{aligned}
  \left.
  \frac{\partial}{\partial\beta^\top}
  \psi_\beta(P_0)
  \right|_{\beta=\beta_0}
  &=
  \left.
  \frac{\partial}{\partial\beta^\top}
  \E_{P_0\otimes P_0}\{K(\beta;\eta_0)\}
  \right|_{\beta=\beta_0} \\
  &\quad +
  \partial_\eta Q(0,\eta_0)
  \left[
  \left.
  \frac{\partial}{\partial\beta^\top}
  \eta_{P_0}^{\mathrm{obs}}(\beta)
  \right|_{\beta=\beta_0}
  \right] \\
  &=
  \left.
  \frac{\partial}{\partial\beta^\top}
  \E_{P_0\otimes P_0}\{K(\beta;\eta_0)\}
  \right|_{\beta=\beta_0},
  \end{aligned}
\]

where the last equality follows from the previously established fact that at $\eta_0$ the first-order nuisance derivative vanishes in every nuisance perturbation direction. Thus the target derivative equals the derivative obtained by holding the true nuisance values fixed at $\eta_0$. Let $\beta_\varepsilon=\beta(P_\varepsilon)$. Differentiating $\psi_{\beta_\varepsilon}(P_\varepsilon)=0$ at zero and using the nonsingularity of this derivative yields
\[
  0
  =
  \left\{
    \left.
    \frac{\partial}{\partial\beta^\top}
    \psi_\beta(P_0)
    \right|_{\beta=\beta_0}
  \right\}\dot\beta
  +
  \E\{D_\psi^{\mathrm{eff}}(O)s(O)\},
  \qquad
  \dot\beta
  =
  \E\!\left[
    -\left\{
      \left.
      \frac{\partial}{\partial\beta^\top}
      \psi_\beta(P_0)
      \right|_{\beta=\beta_0}
    \right\}^{-1}
    D_\psi^{\mathrm{eff}}(O)s(O)
  \right].
\]
Thus
\[
  D_\beta^{\mathrm{eff}}
  =
  -\left\{
    \left.
    \frac{\partial}{\partial\beta^\top}
    \psi_\beta(P_0)
    \right|_{\beta=\beta_0}
  \right\}^{-1}
  D_\psi^{\mathrm{eff}}
  =
  -2\left\{
    \left.
    \frac{\partial}{\partial\beta^\top}
    \psi_\beta(P_0)
    \right|_{\beta=\beta_0}
  \right\}^{-1}
  \kappa
\]
is a gradient of the root functional. It belongs to the tangent-space closure because $D_\psi^{\mathrm{eff}}$ does, and hence it is the canonical gradient, equivalently the efficient influence function, of $\beta(P)$.
\end{proof}

\subsection{Consistency}

\begin{proof}[Proof of \headingcref{prop:estimator-consistency}{Proposition}.]
Evaluating the moment functional $\psi_\beta(P)$ at the true law and substituting the definition of $K$ gives
\[
\begin{aligned}
\psi_\beta(P_0)
&=
\E\!\left[
K\{\beta;\eta_{P_0}^{\mathrm{obs}}(\beta)\}
\right] \\
&=
\E\Bigl[
\omega^{1,0}(e_0)
\bigl\{
S^{\mathrm{FC}}\{\beta;G_0,\Phi_{P_0}^{\mathrm{obs}}(\beta)\}
-
m_{P_0}^{\mathrm{obs}}(X_1,X_2;\beta)
\bigr\}
+
m_{P_0}^{\mathrm{obs}}(X_1,X_2;\beta)
\Bigr].
\end{aligned}
\]

Conditional on $(X_1,X_2)$, independence of the two sampled subjects gives
\[
\begin{aligned}
\E\{\omega^{1,0}(e_0)\mid X_1,X_2\}
&=
\frac{
\E\{A_1(1-A_2)\mid X_1,X_2\}
}{
e_0(X_1)\{1-e_0(X_2)\}
} \\
&=
\frac{
\E(A_1\mid X_1)\,
\E(1-A_2\mid X_2)
}{
e_0(X_1)\{1-e_0(X_2)\}
} \\
&=
\frac{
e_0(X_1)\{1-e_0(X_2)\}
}{
e_0(X_1)\{1-e_0(X_2)\}
}
=
1.
\end{aligned}
\]
Likewise,
\[
\begin{aligned}
&\E\!\left[
\omega^{1,0}(e_0)
S^{\mathrm{FC}}\{\beta;G_0,\Phi_{P_0}^{\mathrm{obs}}(\beta)\}
\mid X_1,X_2
\right] \\
&\quad=
\frac{
\Prb(A_1=1,A_2=0\mid X_1,X_2)
}{
e_0(X_1)\{1-e_0(X_2)\}
}
\E\!\left[
S^{\mathrm{FC}}\{\beta;G_0,\Phi_{P_0}^{\mathrm{obs}}(\beta)\}
\mid X_1,X_2,A_1=1,A_2=0
\right] \\
&\quad=
\E\!\left[
S^{\mathrm{FC}}\{\beta;G_0,\Phi_{P_0}^{\mathrm{obs}}(\beta)\}
\mid X_1,X_2,A_1=1,A_2=0
\right].
\end{aligned}
\]
Since $m_{P_0}^{\mathrm{obs}}(X_1,X_2;\beta)$ is measurable with respect to $(X_1,X_2)$, it follows that
\[
\begin{aligned}
&\E\!\left[
K\{\beta;\eta_{P_0}^{\mathrm{obs}}(\beta)\}
\mid X_1,X_2
\right] \\
&\quad=
\E\!\left[
\omega^{1,0}(e_0)
\left\{
S^{\mathrm{FC}}\{\beta;G_0,\Phi_{P_0}^{\mathrm{obs}}(\beta)\}
-
m_{P_0}^{\mathrm{obs}}(X_1,X_2;\beta)
\right\}
+
m_{P_0}^{\mathrm{obs}}(X_1,X_2;\beta)
\mid X_1,X_2
\right] \\
&\quad=
\E\!\left[
\omega^{1,0}(e_0)
S^{\mathrm{FC}}\{\beta;G_0,\Phi_{P_0}^{\mathrm{obs}}(\beta)\}
\mid X_1,X_2
\right]
-
m_{P_0}^{\mathrm{obs}}(X_1,X_2;\beta)
\E\{\omega^{1,0}(e_0)\mid X_1,X_2\}
+
m_{P_0}^{\mathrm{obs}}(X_1,X_2;\beta) \\
&\quad=
\E\!\left[
\omega^{1,0}(e_0)
S^{\mathrm{FC}}\{\beta;G_0,\Phi_{P_0}^{\mathrm{obs}}(\beta)\}
\mid X_1,X_2
\right]
-
m_{P_0}^{\mathrm{obs}}(X_1,X_2;\beta)
+
m_{P_0}^{\mathrm{obs}}(X_1,X_2;\beta) \\
&\quad=
\E\!\left[
S^{\mathrm{FC}}\{\beta;G_0,\Phi_{P_0}^{\mathrm{obs}}(\beta)\}
\mid X_1,X_2,A_1=1,A_2=0
\right] \\
&\quad=
\E\!\left[
S^\circ(\beta)
\mid X_1,X_2,A_1=1,A_2=0
\right] \\
&\quad=
\frac{
\Prb(A_1=1,A_2=0\mid X_1,X_2)
}{
e_0(X_1)\{1-e_0(X_2)\}
}
\E\!\left[
S^\circ(\beta)
\mid X_1,X_2,A_1=1,A_2=0
\right] \\
&\quad=
\E\{\omega^{1,0}(e_0)S^\circ(\beta)\mid X_1,X_2\}.
\end{aligned}
\]
The fifth equality replacing $S^{\mathrm{FC}}$ by $S^\circ$ follows by carrying each step of the proof of \eqref{eq:censoring-correct-G} conditionally. Taking expectations yields
\[
\psi_\beta(P_0)
=
\E\{\omega^{1,0}(e_0)S^\circ(\beta)\}.
\]
Write $p_\beta=\expit\{\beta^\top\widetilde X\}$. Differentiating the complete-data score gives, for every nonzero $\rho\in\R^{p+1}$,
\[
  -\rho^\top
  \frac{\partial}{\partial\beta^\top}
  \psi_\beta(P_0)
  \rho
  =
  \E\!\left[
    \omega^{1,0}(e_0)
    \int_0^\tau
    R^{\circ}(t)\,
    p_\beta(1-p_\beta)
    \bigl(\rho^\top\widetilde X\bigr)^2
    \,dt
  \right]
  \ge 0,
\]
where the last inequality follows because $\omega^{1,0}(e_0)\ge 0$, $R^{\circ}(t)\ge 0$, $p_\beta(1-p_\beta)>0$, and $\bigl(\rho^\top\widetilde X\bigr)^2\ge 0$. At $\beta_0$, the nonsingularity of
\[
  \left.
  \frac{\partial}{\partial\beta^\top}
  \psi_\beta(P_0)
  \right|_{\beta=\beta_0}
\]
under \Cref{ass:kernel-regularity} makes the right-hand side strictly positive for every nonzero $\rho$. Note that
\[
  -\rho^\top
  \frac{\partial}{\partial\beta^\top}
  \psi_\beta(P_0)
  \rho
\]
depends on $\beta$ only through the strictly positive factor $p_\beta(1-p_\beta)$. Hence \eqref{eq:target-unique-cond} holds for every $\beta$, and \Cref{prop:target-unique} shows that $\beta_0$ is the unique zero of $\beta\mapsto\psi_\beta(P_0)$ on $\Theta$.

For any fixed $\beta\in\Theta$, recall that $\widehat U_n(\beta)=\{n(n-1)\}^{-1} \sum_{i\ne j}K_{ij}(\beta;\widehat\eta_{ij})$. Let $I_1,\ldots,I_F$ denote the subject-level folds used for cross-fitting, and let $\widehat\eta^{-(f,g)}$ denote the nuisance fit based on subjects outside $I_f\cup I_g$. Define
\[
  \mathcal B_{fg}
  =
  \{(i,j):i\in I_f,\ j\in I_g,\ i\ne j\},
\]
and let $\mathcal T_{fg}$ be the corresponding training sigma-field.

For each fixed fold pair $(f,g)$, the assumed $L_2(P_0)$ consistency gives
\[
\begin{aligned}
  &\|\widehat e^{-(f,g)}-e_0\|_2
  +\|\widehat G^{-(f,g)}-G_0\|_2
  \\
  &\qquad+
  \|\widehat\Phi^{-(f,g)}(\beta)
    -\Phi_{P_0}^{\mathrm{obs}}(\beta)\|_2
  +
  \|\widehat m^{-(f,g)}(\beta)
    -m_{P_0}^{\mathrm{obs}}(\beta)\|_2
  =
  o_p(1).
\end{aligned}
\]
The local Lipschitz condition in \Cref{ass:kernel-regularity} therefore yields
\begin{equation}
\label{eq:crossfit-kernel-l2}
  \left\|
    K\{\beta;\widehat\eta^{-(f,g)}(\beta)\}
    -
    K\{\beta;\eta_{P_0}^{\mathrm{obs}}(\beta)\}
  \right\|_{L_2(P_0\otimes P_0)}
  =
  o_p(1).
\end{equation}

The subjects in $I_f\cup I_g$ are independent of the training observations that generate $\mathcal T_{fg}$. Thus, conditioning on $\mathcal T_{fg}$ fixes the nuisance fit but does not change the distribution of an evaluation pair $(O_i,O_j)$, which remains $P_0\otimes P_0$. Hence
\[
\begin{aligned}
&\E\!\left[
\left.
\left\|
\frac{1}{|\mathcal B_{fg}|}
\sum_{(i,j)\in\mathcal B_{fg}}
\left[
K_{ij}\{\beta;\widehat\eta^{-(f,g)}(\beta)\}
-
K_{ij}\{\beta;\eta_{P_0}^{\mathrm{obs}}(\beta)\}
\right]
\right\|
\,\right|\,\mathcal T_{fg}
\right]
\\
&\quad\le
\frac{1}{|\mathcal B_{fg}|}
\sum_{(i,j)\in\mathcal B_{fg}}
\E\!\left[
\left.
\left\|
K_{ij}\{\beta;\widehat\eta^{-(f,g)}(\beta)\}
-
K_{ij}\{\beta;\eta_{P_0}^{\mathrm{obs}}(\beta)\}
\right\|
\,\right|\,\mathcal T_{fg}
\right]
\\
&\quad=
\E_{P_0\otimes P_0}
\left\|
K\{\beta;\widehat\eta^{-(f,g)}(\beta)\}
-
K\{\beta;\eta_{P_0}^{\mathrm{obs}}(\beta)\}
\right\|
\\
&\quad\le
\left\|
K\{\beta;\widehat\eta^{-(f,g)}(\beta)\}
-
K\{\beta;\eta_{P_0}^{\mathrm{obs}}(\beta)\}
\right\|_{L_2(P_0\otimes P_0)}
=
o_p(1).
\end{aligned}
\]
Markov's inequality therefore gives
\[
\frac{1}{|\mathcal B_{fg}|}
\sum_{(i,j)\in\mathcal B_{fg}}
\left[
K_{ij}\{\beta;\widehat\eta^{-(f,g)}(\beta)\}
-
K_{ij}\{\beta;\eta_{P_0}^{\mathrm{obs}}(\beta)\}
\right]
=
o_p(1).
\]

Each ordered pair $(i,j)$ belongs to exactly one $\mathcal B_{fg}$, determined by the folds containing $i$ and $j$. Since $F$ is fixed,
\[
\begin{aligned}
&\widehat U_n(\beta)
-
\frac{1}{n(n-1)}
\sum_{i\ne j}
K_{ij}\{\beta;\eta_{P_0}^{\mathrm{obs}}(\beta)\}
\\
&\quad=
\sum_{f,g}
\frac{|\mathcal B_{fg}|}{n(n-1)}
\frac{1}{|\mathcal B_{fg}|}
\sum_{(i,j)\in\mathcal B_{fg}}
\left[
K_{ij}\{\beta;\widehat\eta^{-(f,g)}(\beta)\}
-
K_{ij}\{\beta;\eta_{P_0}^{\mathrm{obs}}(\beta)\}
\right]
\\
&\quad=
o_p(1).
\end{aligned}
\]

The U-statistic law of large numbers now gives
\[
\begin{aligned}
\widehat U_n(\beta)
&=
\E_{P_0\otimes P_0}
\left[
K_{12}\{\beta;\eta_{P_0}^{\mathrm{obs}}(\beta)\}
\right]
+o_p(1)
\\
&=
\psi_\beta(P_0)+o_p(1).
\end{aligned}
\]
The $L_2$-Lipschitz conditions for the population and fitted kernel paths in \Cref{ass:kernel-regularity,ass:crossfit-regularity} control the change in both scores as $\beta$ varies over the compact set $\Theta$. The pointwise convergence above therefore extends uniformly:
\[
  \sup_{\beta\in\Theta}
  \|\widehat U_n(\beta)-\psi_\beta(P_0)\|
  =
  o_p(1).
\]

As $\beta_0$ is the unique zero of $\beta\mapsto\psi_\beta(P_0)$ on $\Theta$ and the population score is continuous in $\beta$, the population score is bounded away from zero outside every neighborhood of $\beta_0$: for each $\varepsilon>0$,
\[
  \inf_{\substack{\beta\in\Theta\\
                  \|\beta-\beta_0\|\ge\varepsilon}}
  \|\psi_\beta(P_0)\|
  >0.
\]
The approximate-root condition and the preceding uniform convergence give
\[
\begin{aligned}
  \|\psi_{\widehat\beta}(P_0)\|
  &=
  \left\|
    \{\psi_{\widehat\beta}(P_0)
    -\widehat U_n(\widehat\beta)\}
    +\widehat U_n(\widehat\beta)
  \right\|
  \\
  &\le
  \|\widehat U_n(\widehat\beta)\|
  +
  \sup_{\beta\in\Theta}
  \|\widehat U_n(\beta)-\psi_\beta(P_0)\|
  \\
  &=
  o_p(1).
\end{aligned}
\]
Consequently,
\[
  \widehat\beta\xrightarrow{p}\beta_0.
\]
\end{proof}

\subsection{Asymptotic Linearity and Normality}
\label{app:asymptotic-linearity}

\begin{lemma}[Monte Carlo Approximation Error]
\label{lem:mc-error}
Let $\widehat\Phi^{\mathrm{fit}}$ be the conditional projection under the fitted future distributions, and let $\widehat\Phi^B$ be defined by \eqref{eq:mc-projection}. Conditional on the observed data and fitted nuisance functions, suppose the Monte Carlo approximation errors entering the average over ordered subject pairs are independent, have conditional mean zero, and have uniformly bounded conditional second moments after multiplication by the score weights. Then, for fixed $B\ge1$,
\[
  \frac{1}{n(n-1)}
  \sum_{i\ne j}
  \left\{
  K_{ij}(\beta_0;\widehat e,\widehat G,\widehat\Phi^B,\widehat m)
  -
  K_{ij}(\beta_0;\widehat e,\widehat G,
  \widehat\Phi^{\mathrm{fit}},\widehat m)
  \right\}
  =
  O_p(n^{-1})
  =
  o_p(n^{-1/2}).
\]
\end{lemma}

\begin{proof}[Proof of \headingcref{lem:mc-error}{Lemma}.]

Define
\[
  \Delta_{ij,B}
  =
  K_{ij}(\beta_0;
    \widehat e,\widehat G,\widehat\Phi^B,\widehat m)
  -
  K_{ij}(\beta_0;
    \widehat e,\widehat G,\widehat\Phi^{\mathrm{fit}},\widehat m),
\]
and let
\[
\mathcal E_{n,B}^{\mathrm{MC}}
=
n^{-1}(n-1)^{-1}\sum_{i\ne j}\Delta_{ij,B}
\]
denote its average. All expectations in this proof are conditional on the observed data and fitted nuisance functions. For distinct ordered pairs, independence and mean zero imply
\[
\E\!\left(\Delta_{ij,B}^{\top}\Delta_{kl,B}\right)=0,
\qquad (i,j)\ne(k,l).
\]
Therefore,
\[
\begin{aligned}
\E\{\|\mathcal E_{n,B}^{\mathrm{MC}}\|^2\}
&=
\frac{1}{n^2(n-1)^2}
\sum_{i\ne j}\sum_{k\ne l}
\E\!\left(\Delta_{ij,B}^{\top}\Delta_{kl,B}\right)\\
&=
\frac{1}{n^2(n-1)^2}
\sum_{i\ne j}\E\{\|\Delta_{ij,B}\|^2\}\\
&\le
\frac{1}{n(n-1)}
\sup_{i\ne j}\E\{\|\Delta_{ij,B}\|^2\}\\
&=
O(n^{-2}).
\end{aligned}
\]
Markov's inequality therefore yields $\mathcal E_{n,B}^{\mathrm{MC}}=O_p(n^{-1})=o_p(n^{-1/2})$ for fixed $B$.
\end{proof}

\cref{lem:mc-error} covers the usual implementation in which Monte Carlo draws are generated independently for each ordered subject pair. When these draws are sampled from the fitted future distribution, each pairwise Monte Carlo average is centered at the corresponding fitted projection, and the simulation errors average out across pairs.

\begin{proof}[Proof of \headingcref{thm:asymptotic-linearity}{Theorem}.]

With the nuisance functions fixed at $\eta_0$,
\[
  U_n(\beta_0;\eta_0)
  =
  \frac{1}{n(n-1)}\sum_{i\ne j}K_{ij}(\beta_0;\eta_0)
\]
is a second-order U-statistic with symmetrized kernel $\bar K$. Hoeffding's decomposition and the kernel second-moment condition yield
\[
  U_n(\beta_0;\eta_0)
  =
  \frac{2}{n}\sum_{i=1}^{n}\kappa(O_i)+o_p(n^{-1/2}),
\]
because the degenerate component has variance of order $n^{-2}$.

Let
\[
  \bar H_{fg}
  =
  |\mathcal B_{fg}|^{-1}
  \sum_{(i,j)\in\mathcal B_{fg}}
  \left\{
    K_{ij}(\beta_0;\widehat\eta^{-(f,g)})
    -
    K_{ij}(\beta_0;\eta_0)
  \right\}.
\]
Conditional on $\mathcal T_{fg}$, $\bar H_{fg}$ is a U-statistic formed by averaging the pairwise kernel difference over $\mathcal B_{fg}$. Define its conditional mean by
\[
  \Delta_{fg}
  =
  \E(\bar H_{fg}\mid\mathcal T_{fg}).
\]
Conditional on $\mathcal T_{fg}$, $\widehat\eta^{-(f,g)}$ is fixed. Because $\mathcal T_{fg}$ is generated by observations from subjects outside $I_f\cup I_g$, conditioning on $\mathcal T_{fg}$ does not change the joint distribution of $\{O_i:i\in I_f\cup I_g\}$. Therefore,
\[
  \Delta_{fg}
  =
  \E_{P_0\otimes P_0}\!\left[
    K_{12}(\beta_0;\widehat\eta^{-(f,g)})
    -
    K_{12}(\beta_0;\eta_0)
  \right].
\]
Applying the product-remainder bound in \Cref{thm:aipw-dr} gives
\[
\begin{aligned}
  \|\Delta_{fg}\|
  =
  O_p\!\Bigl(
    &\|\widehat e^{-(f,g)}-e_0\|_2\,
     \|\widehat m^{-(f,g)}-m_0\|_2
    \\
    &+
    \|\widehat G^{-(f,g)}-G_0\|_2\,
     \|\widehat\Phi^{-(f,g)}-\Phi_0\|_2
  \Bigr).
\end{aligned}
\]
Because the product-rate conditions hold uniformly over the fixed number of fold pairs,
\begin{equation}
\label{eq:block-drift-rate}
  \max_{f,g}\|\Delta_{fg}\|
  =
  o_p(n^{-1/2}).
\end{equation}

Conditional on $\mathcal T_{fg}$, $\bar H_{fg}-\Delta_{fg}$ is a centered U-statistic. The second-moment bound from the Hoeffding decomposition therefore gives
\[
  \E\!\left\{
    \|\bar H_{fg}-\Delta_{fg}\|^2
    \mid\mathcal T_{fg}
  \right\}
  =
  O\!\left(
    \frac{1}{n}
    \left\|
      K_{12}(\beta_0;\widehat\eta^{-(f,g)})
      -
      K_{12}(\beta_0;\eta_0)
    \right\|_2^2
  \right)
  =
  o_p(n^{-1}),
\]
where the last equality follows from \eqref{eq:crossfit-kernel-l2} at $\beta=\beta_0$ together with \eqref{eq:truth-law-nuisance}. Conditional Markov's inequality and the fixed number of fold pairs then give
\begin{equation}
\label{eq:block-fluctuation-rate}
  \max_{f,g}
  \|\bar H_{fg}-\Delta_{fg}\|
  =
  o_p(n^{-1/2}).
\end{equation}
Combining \eqref{eq:block-drift-rate} and \eqref{eq:block-fluctuation-rate} gives
\[
  \max_{f,g}\|\bar H_{fg}\|
  =
  o_p(n^{-1/2}).
\]
Averaging these finitely many blocks,
\[
  \frac{1}{n(n-1)}
  \sum_{i\ne j}
  \{K_{ij}(\beta_0;\widehat\eta^{-(f_i,f_j)})
    -K_{ij}(\beta_0;\eta_0)\}
  =o_p(n^{-1/2}),
\]
where $f_i$ is the fold of subject $i$. Hence
\[
  \widehat U_n(\beta_0)
  =
  \frac{2}{n}\sum_{i=1}^{n}\kappa(O_i)+o_p(n^{-1/2}).
\]
Replacing $\widehat\Phi^{\mathrm{det}}$ by $\widehat\Phi^B$ adds only an $o_p(n^{-1/2})$ perturbation to the estimating equation.

Let
\[
  \widetilde J_n
  =
  \int_0^1
  \frac{\partial}{\partial\beta^\top}
  \widehat U_n\!\left\{
    \beta_0+t(\widehat\beta-\beta_0)
  \right\}\,dt.
\]
Adding and subtracting the population derivative inside the integral gives
\[
\begin{aligned}
  \widetilde J_n
  -
  \left.
  \frac{\partial}{\partial\beta^\top}
  \psi_\beta(P_0)
  \right|_{\beta=\beta_0}
  &=
  \int_0^1
  \left[
    \frac{\partial}{\partial\beta^\top}
    \widehat U_n\!\left\{
      \beta_0+t(\widehat\beta-\beta_0)
    \right\}
    -
    \left.
    \frac{\partial}{\partial\beta^\top}
    \psi_\beta(P_0)
    \right|_{\beta=\beta_0+t(\widehat\beta-\beta_0)}
  \right]\,dt
  \\
  &\quad+
  \int_0^1
  \left[
    \left.
    \frac{\partial}{\partial\beta^\top}
    \psi_\beta(P_0)
    \right|_{\beta=\beta_0+t(\widehat\beta-\beta_0)}
    -
    \left.
    \frac{\partial}{\partial\beta^\top}
    \psi_\beta(P_0)
    \right|_{\beta=\beta_0}
  \right]\,dt.
\end{aligned}
\]
Because $\widehat{\beta} \xrightarrow{p} \beta_0$, the first integral is $o_p(1)$ by Assumption~\ref{ass:crossfit-regularity}, and the second is $o_p(1)$ by the continuous differentiability in Assumption~\ref{ass:kernel-regularity}. Thus
\[
  \widetilde J_n
  =
  \left.
  \frac{\partial}{\partial\beta^\top}
  \psi_\beta(P_0)
  \right|_{\beta=\beta_0}
  +o_p(1),
\]
and
\[
  \widetilde J_n^{-1}
  =
  \left\{
    \left.
    \frac{\partial}{\partial\beta^\top}
    \psi_\beta(P_0)
    \right|_{\beta=\beta_0}
  \right\}^{-1}
  +o_p(1).
\]
Taylor's formula gives
\[
  \widehat U_n(\widehat\beta)-\widehat U_n(\beta_0)
  =
  \widetilde J_n(\widehat\beta-\beta_0).
\]
At the approximate root, $\widehat U_n(\widehat\beta)=o_p(n^{-1/2})$. Substituting this then yields
\[
  \sqrt n(\widehat\beta-\beta_0)
  =
  -\widetilde J_n^{-1}\sqrt n\,\widehat U_n(\beta_0)+o_p(1)
  =
  -\left\{
    \left.
    \frac{\partial}{\partial\beta^\top}
    \psi_\beta(P_0)
    \right|_{\beta=\beta_0}
  \right\}^{-1}
  \sqrt n\,\widehat U_n(\beta_0)+o_p(1).
\]

Combining the preceding display with the Hoeffding expansion and \eqref{eq:beta-eif} gives
\[
  \sqrt n(\widehat\beta-\beta_0)
  =
  \frac{1}{\sqrt n}
  \sum_{i=1}^n D_\beta^{\mathrm{eff}}(O_i)
  +o_p(1).
\]
The central limit theorem for the independent first Hoeffding projections establishes the stated normal limit.

\end{proof}

The rate requirements in \Cref{thm:asymptotic-linearity} are attainable under standard conditions for the nuisance estimators considered here. For correctly specified fixed-dimensional parametric propensity and censoring models, such as logistic regression, standard parametric theory gives root-$n$ consistency for $\widehat e$ and $\widehat G$. Because the number of folds is fixed and each training sample contains a fixed proportion of the subjects, these rates hold uniformly over the nuisance fits.

In our implementations, both $\widehat\Phi^{\mathrm{det}}$ and $\widehat m$ are evaluated by deterministic recursion under fitted transition models. On the finite state space and fixed visit grid, these quantities are finite sums and products of the fitted transition probabilities. Consistency of the fitted transition probabilities therefore implies
\[
  \|\widehat\Phi^{\mathrm{det}}-\Phi_0\|=o_p(1),
  \qquad
  \|\widehat m-m_0\|=o_p(1)
\]
by the continuous mapping theorem.

\noindent\textbf{Alternative direct regression estimation of $\Phi_0$ and $m_0$.}
For each $\beta$, define
\[
  \Psi_{\ell}^{\mathrm{IPW}}(\beta;G)
  =
  \sum_{r>\ell}
  \frac{Y_{r-1}G_{\ell-1}}
       {G_{r-1}}
  S_{r}^{\circ}(\beta).
\]
The at-risk indicator retains an interval-$r$ score contribution only when both subjects remain observed through $t_{r-1}$, and the survival ratio reweights that contribution by the inverse conditional probability of remaining uncensored from the pair history at interval $\ell$ through $t_{r-1}$. Sequentially independent censoring therefore gives
\begin{equation}
\label{eq:phi-regression-conditional-mean}
  \E\!\left\{
    \Psi_{\ell}^{\mathrm{IPW}}(\beta_0;G_0)
    \mid \mathcal H_{\ell}^{+}
  \right\}
  =
  \Phi_{0,\ell}(\beta_0).
\end{equation}
Let $\widehat\Phi_{\ell}^{\mathrm{reg}}(\beta)$ be the least-squares regression, under squared Euclidean loss, of the pairwise values $\Psi_{ij,\ell}^{\mathrm{IPW}}(\beta;\widehat G)$ on $\mathcal H_{ij,\ell}^{+}$ over a function class $\mathcal F_\ell$.

Given $\widehat\Phi^{\mathrm{reg}}(\beta)$, let $\widehat m^{\mathrm{reg}}(\cdot,\cdot;\beta)$ be the least-squares regression, also under squared Euclidean loss, of
\[
  S_{ij}^{\mathrm{FC}}
  \bigl(\beta;\widehat G,\widehat\Phi^{\mathrm{reg}}\bigr)
\]
on $(X_i,X_j)$ among treated-versus-control pairs over a function class $\mathcal M$. The conditional identification result following \Cref{thm:future-score-dr} gives
\begin{equation}
\label{eq:m-regression-conditional-mean}
  \E\!\left\{
    S_{ij}^{\mathrm{FC}}(\beta_0;G_0,\Phi_0)
    \mid X_i,X_j,A_i=1,A_j=0
  \right\}
  =
  m_0(X_i,X_j;\beta_0).
\end{equation}
Equations~\eqref{eq:phi-regression-conditional-mean} and \eqref{eq:m-regression-conditional-mean} identify $\Phi_{0,\ell}(\beta_0)$ and $m_0(\cdot,\cdot;\beta_0)$ as the targets of the two regressions.

Suppose $\mathcal F_\ell$ contains $\Phi_{0,\ell}(\beta_0)$ and $\mathcal M$ contains $m_0(\cdot,\cdot;\beta_0)$. Suppose the following convergences hold uniformly over the training samples:
\begin{align}
\label{eq:phi-pseudo-outcome-l2}
  \bigl\|
    \Psi_{\ell}^{\mathrm{IPW}}(\beta_0;\widehat G)
    -
    \Psi_{\ell}^{\mathrm{IPW}}(\beta_0;G_0)
  \bigr\|_2
  &=o_p(1),\\
\label{eq:m-pseudo-outcome-l2}
  \bigl\|
    S^{\mathrm{FC}}
    \bigl(\beta_0;\widehat G,\widehat\Phi^{\mathrm{reg}}\bigr)
    -
    S^{\mathrm{FC}}(\beta_0;G_0,\Phi_0)
  \bigr\|_2
  &=o_p(1).
\end{align}

Suppose further that the squared Euclidean loss classes used to fit $\widehat\Phi_{\ell}^{\mathrm{reg}}$ over $\mathcal F_\ell$ and $\widehat m^{\mathrm{reg}}$ over $\mathcal M$ are bounded and U-Glivenko--Cantelli. Hence, for each regression, the empirical average of the squared loss converges uniformly to its expectation \citep{arcones1993limit}. Since the fitted regressions minimize these empirical squared losses, the preceding uniform convergence, together with \eqref{eq:phi-pseudo-outcome-l2}, \eqref{eq:m-pseudo-outcome-l2}, \eqref{eq:phi-regression-conditional-mean}, and \eqref{eq:m-regression-conditional-mean}, implies that the expected squared loss induced by each fitted regression converges to the expected squared loss induced by the corresponding target function. This gives
\[
  \|\widehat\Phi_{\ell}^{\mathrm{reg}}(\beta_0)
    -\Phi_{0,\ell}(\beta_0)\|_2=o_p(1),
  \qquad
  \|\widehat m^{\mathrm{reg}}(\cdot,\cdot;\beta_0)
    -m_0(\cdot,\cdot;\beta_0)\|_2=o_p(1).
\]
Because the visit grid is finite,
\[
  \|\widehat\Phi^{\mathrm{reg}}(\beta_0)
    -\Phi_0(\beta_0)\|_2=o_p(1).
\]

Together with the root-$n$ rates for $\widehat e$ and $\widehat G$, these consistency results imply the product-rate requirements in \Cref{thm:asymptotic-linearity}, with $\widehat\Phi=\widehat\Phi^{\mathrm{reg}}$ and $\widehat m=\widehat m^{\mathrm{reg}}$.

\begin{proof}[Proof of \headingcref{cor:wald-inference}{Corollary}.]

Let
\[
  \widehat\kappa_i^\ast
  =
  \frac{1}{n-1}\sum_{j\ne i}\bar K(O_i,O_j)
  -
  \frac{1}{n}\sum_{r=1}^n
  \frac{1}{n-1}\sum_{s\ne r}\bar K(O_r,O_s)
\]
be the centered empirical estimate of the first Hoeffding projection $ \kappa(O_i) = \E\!\left\{\bar K(O_i,O_2)\mid O_i\right\}. $ Let
\[
\begin{aligned}
  \widehat\kappa_i^{\mathrm{fit}}
  &=
  \frac{1}{n-1}\sum_{j\ne i}
  \frac{1}{2}
  \Big\{
    K_{ij}(\widehat\beta;
      \widehat e,\widehat G,\widehat\Phi^{\mathrm{fit}},\widehat m)
    +
    K_{ji}(\widehat\beta;
      \widehat e,\widehat G,\widehat\Phi^{\mathrm{fit}},\widehat m)
  \Big\}
  \\
  &\quad-
  \frac{1}{n}\sum_{r=1}^n
  \frac{1}{n-1}\sum_{s\ne r}
  \frac{1}{2}
  \Big\{
    K_{rs}(\widehat\beta;
      \widehat e,\widehat G,\widehat\Phi^{\mathrm{fit}},\widehat m)
    +
    K_{sr}(\widehat\beta;
      \widehat e,\widehat G,\widehat\Phi^{\mathrm{fit}},\widehat m)
  \Big\}.
\end{aligned}
\]
Thus, $\widehat\kappa_i^{\mathrm{fit}}$ replaces the oracle pair kernel $\bar K(O_i,O_j)$ in $\widehat\kappa_i^\ast$ by its fitted counterpart evaluated at $\widehat\beta$. The law of large numbers for projections of second-order U-statistics yields
\[
  \frac{1}{n}\sum_{i=1}^n
  \widehat\kappa_i^\ast\widehat\kappa_i^{\ast\top}
  \xrightarrow{p}
  \Var\{\kappa(O)\}.
\]

For \(i\ne j\), define
\[
  d_{ij}
  =
  \frac{1}{2}
  \Big\{
    K_{ij}(\widehat\beta;
      \widehat e,\widehat G,\widehat\Phi^{\mathrm{fit}},\widehat m)
    +
    K_{ji}(\widehat\beta;
      \widehat e,\widehat G,\widehat\Phi^{\mathrm{fit}},\widehat m)
  \Big\}
  -
  \bar K(O_i,O_j).
\]
Subtracting the two empirical projections gives
\[
  \widehat\kappa_i^{\mathrm{fit}}
  -
  \widehat\kappa_i^\ast
  =
  \frac{1}{n-1}\sum_{j\ne i}d_{ij}
  -
  \frac{1}{n(n-1)}\sum_{r\ne s}d_{rs}.
\]
The second term is the sample mean of the first term over \(i\), since
\[
  \frac{1}{n}
  \sum_{i=1}^n
  \left(
    \frac{1}{n-1}\sum_{j\ne i}d_{ij}
  \right)
  =
  \frac{1}{n(n-1)}
  \sum_{r\ne s}d_{rs}.
\]
Consequently,
\[
\begin{aligned}
  &\frac{1}{n}\sum_{i=1}^n
  \Biggl\|
    \widehat\kappa_i^{\mathrm{fit}}
    -
    \widehat\kappa_i^\ast
  \Biggr\|^2
  \\
  &=
  \frac{1}{n}\sum_{i=1}^n
  \Biggl\|
    \frac{1}{n-1}\sum_{j\ne i}d_{ij}
    -
    \frac{1}{n(n-1)}\sum_{r\ne s}d_{rs}
  \Biggr\|^2
  \\
  &=
  \frac{1}{n}\sum_{i=1}^n
  \Biggl[
    \Biggl\|
      \frac{1}{n-1}\sum_{j\ne i}d_{ij}
    \Biggr\|^2
    -
    2
    \left(
      \frac{1}{n-1}\sum_{j\ne i}d_{ij}
    \right)^\top
    \left(
      \frac{1}{n(n-1)}\sum_{r\ne s}d_{rs}
    \right)
    +
    \Biggl\|
      \frac{1}{n(n-1)}\sum_{r\ne s}d_{rs}
    \Biggr\|^2
  \Biggr]
  \\
  &=
  \frac{1}{n}\sum_{i=1}^n
  \Biggl\|
    \frac{1}{n-1}\sum_{j\ne i}d_{ij}
  \Biggr\|^2
  -
  \Biggl\|
    \frac{1}{n(n-1)}\sum_{r\ne s}d_{rs}
  \Biggr\|^2
  \\
  &\le
  \frac{1}{n}\sum_{i=1}^n
  \Biggl\|
    \frac{1}{n-1}\sum_{j\ne i}d_{ij}
  \Biggr\|^2.
\end{aligned}
\]
For each \(i\), Jensen's inequality gives
\[
  \Biggl\|
    \frac{1}{n-1}\sum_{j\ne i}d_{ij}
  \Biggr\|^2
  \le
  \frac{1}{n-1}\sum_{j\ne i}\|d_{ij}\|^2.
\]
Therefore,
\[
\begin{aligned}
  &\frac{1}{n}\sum_{i=1}^n
  \Biggl\|
    \widehat\kappa_i^{\mathrm{fit}}
    -
    \widehat\kappa_i^\ast
  \Biggr\|^2
  \\
  &\le
  \frac{1}{n(n-1)}
  \sum_{i\ne j}
  \Biggl\|
    \frac{1}{2}
    \Big\{
      K_{ij}(\widehat\beta;
        \widehat e,\widehat G,\widehat\Phi^{\mathrm{fit}},\widehat m)
      +
      K_{ji}(\widehat\beta;
        \widehat e,\widehat G,\widehat\Phi^{\mathrm{fit}},\widehat m)
    \Big\}
    -
    \bar K(O_i,O_j)
  \Biggr\|^2.
\end{aligned}
\]

For the following argument, write the dependence of the fitted nuisance functions on $\beta$ explicitly:
\[
  K_{ij}(\beta;
    \widehat e,\widehat G,\widehat\Phi^{\mathrm{fit}},\widehat m)
  =
  K_{ij}\!\left\{
    \beta;
    \widehat e,\widehat G,
    \widehat\Phi^{\mathrm{fit}}(\beta),
    \widehat m(\cdot,\cdot;\beta)
  \right\}.
\]

Inside each norm on the right, add and subtract
\[
  \frac{1}{2}
  \Big[
    K_{ij}\!\left\{
      \beta_0;
      \widehat e,\widehat G,
      \widehat\Phi^{\mathrm{fit}}(\beta_0),
      \widehat m(\cdot,\cdot;\beta_0)
    \right\}
    +
    K_{ji}\!\left\{
      \beta_0;
      \widehat e,\widehat G,
      \widehat\Phi^{\mathrm{fit}}(\beta_0),
      \widehat m(\cdot,\cdot;\beta_0)
    \right\}
  \Big],
\]
and use $\|a+b\|^2\leq 2\|a\|^2+2\|b\|^2$ to obtain
\begin{equation}
\label{eq:projection-error-bound}
\begin{aligned}
  &\frac{1}{n}\sum_{i=1}^n
  \Biggl\|
    \widehat\kappa_i^{\mathrm{fit}}
    -
    \widehat\kappa_i^\ast
  \Biggr\|^2
  \\
  &\leq
  \frac{2}{n(n-1)}
  \sum_{i\ne j}
  \Biggl\|
    \frac{1}{2}
    \Biggl[
      K_{ij}(\widehat\beta;
        \widehat e,\widehat G,
        \widehat\Phi^{\mathrm{fit}},\widehat m)
      -
      K_{ij}(\beta_0;
        \widehat e,\widehat G,
        \widehat\Phi^{\mathrm{fit}},\widehat m)
  \\
  &\hspace{12em}
      +
      K_{ji}(\widehat\beta;
        \widehat e,\widehat G,
        \widehat\Phi^{\mathrm{fit}},\widehat m)
      -
      K_{ji}(\beta_0;
        \widehat e,\widehat G,
        \widehat\Phi^{\mathrm{fit}},\widehat m)
    \Biggr]
  \Biggr\|^2
  \\
  &\quad+
  \frac{2}{n(n-1)}
  \sum_{i\ne j}
  \Biggl\|
    \frac{1}{2}
    \Big\{
      K_{ij}(\beta_0;
        \widehat e,\widehat G,
        \widehat\Phi^{\mathrm{fit}},\widehat m)
      +
      K_{ji}(\beta_0;
        \widehat e,\widehat G,
        \widehat\Phi^{\mathrm{fit}},\widehat m)
    \Big\}
    -
    \bar K(O_i,O_j)
  \Biggr\|^2.
\end{aligned}
\end{equation}
By \Cref{ass:crossfit-regularity} and $\widehat\beta\xrightarrow{p}\beta_0$,
\[
\begin{aligned}
&\frac{1}{n(n-1)}
\sum_{i\ne j}
\left\|
K_{ij}(\widehat\beta;
  \widehat e,\widehat G,
  \widehat\Phi^{\mathrm{fit}},\widehat m)
-
K_{ij}(\beta_0;
  \widehat e,\widehat G,
  \widehat\Phi^{\mathrm{fit}},\widehat m)
\right\|^2
\\
&\qquad=
O_p\!\left(\|\widehat\beta-\beta_0\|^2\right)
=
o_p(1).
\end{aligned}
\]
The same bound holds with $i$ and $j$ interchanged. Therefore,
\begin{equation}
\label{eq:beta-perturbation-bound}
\begin{aligned}
  &\frac{1}{n(n-1)}
  \sum_{i\ne j}
  \Biggl\|
    \frac{1}{2}
    \Biggl[
      K_{ij}(\widehat\beta;
        \widehat e,\widehat G,
        \widehat\Phi^{\mathrm{fit}},\widehat m)
      -
      K_{ij}(\beta_0;
        \widehat e,\widehat G,
        \widehat\Phi^{\mathrm{fit}},\widehat m)
  \\
  &\hspace{12em}
      +
      K_{ji}(\widehat\beta;
        \widehat e,\widehat G,
        \widehat\Phi^{\mathrm{fit}},\widehat m)
      -
      K_{ji}(\beta_0;
        \widehat e,\widehat G,
        \widehat\Phi^{\mathrm{fit}},\widehat m)
    \Biggr]
  \Biggr\|^2
  \\
  &=
  O_p\!\left(\|\widehat\beta-\beta_0\|^2\right)
  =
  o_p(1).
\end{aligned}
\end{equation}

Recall from the proof of \Cref{prop:estimator-consistency} that $I_1,\ldots,I_F$ are the subject-level folds, with $F$ fixed, and that for $f,g\in\{1,\ldots,F\}$, $ \mathcal B_{fg} = \{(i,j):i\in I_f,\ j\in I_g,\ i\ne j\}. $ The sigma-field $\mathcal T_{fg}$ is generated by the training observations outside $I_f\cup I_g$, where the nuisance functions used for pairs in $\mathcal B_{fg}$ are fitted.

At $\beta_0$, the conditional kernel convergence established in that proof gives, for $(i,j)\in\mathcal B_{fg}$,
\[
\begin{aligned}
  &\E\!\Biggl[
    \Biggl\|
      K_{ij}(\beta_0;
        \widehat e,\widehat G,
        \widehat\Phi^{\mathrm{fit}},\widehat m)
      -
      K_{ij}(\beta_0;\eta_0)
    \Biggr\|^2
    \,\Biggm|\,
    \mathcal T_{fg}
  \Biggr]
  =
  o_p(1),
  \\
  &\E\!\Biggl[
    \Biggl\|
      K_{ji}(\beta_0;
        \widehat e,\widehat G,
        \widehat\Phi^{\mathrm{fit}},\widehat m)
      -
      K_{ji}(\beta_0;\eta_0)
    \Biggr\|^2
    \,\Biggm|\,
    \mathcal T_{fg}
  \Biggr]
  =
  o_p(1).
\end{aligned}
\]
Since
\[
  \bar K(O_i,O_j)
  =
  \frac{1}{2}
  \left\{
    K_{ij}(\beta_0;\eta_0)
    +
    K_{ji}(\beta_0;\eta_0)
  \right\},
\]
conditioning on $\mathcal T_{fg}$ and using
\[
  \Biggl\|
    \frac{a+b}{2}
  \Biggr\|^2
  \leq
  \frac{1}{2}\|a\|^2
  +
  \frac{1}{2}\|b\|^2
\]
give
\[
\begin{aligned}
  &\E\!\Biggl[
    \Biggl\|
      \frac{1}{2}
      \Big\{
        K_{ij}(\beta_0;
          \widehat e,\widehat G,
          \widehat\Phi^{\mathrm{fit}},\widehat m)
        +
        K_{ji}(\beta_0;
          \widehat e,\widehat G,
          \widehat\Phi^{\mathrm{fit}},\widehat m)
      \Big\}
      -
      \bar K(O_i,O_j)
    \Biggr\|^2
    \,\Biggm|\,
    \mathcal T_{fg}
  \Biggr]
  \\
  &\leq
  \frac{1}{2}
  \E\!\Biggl[
    \Biggl\|
      K_{ij}(\beta_0;
        \widehat e,\widehat G,
        \widehat\Phi^{\mathrm{fit}},\widehat m)
      -
      K_{ij}(\beta_0;\eta_0)
    \Biggr\|^2
    \,\Biggm|\,
    \mathcal T_{fg}
  \Biggr]
  \\
  &\quad+
  \frac{1}{2}
  \E\!\Biggl[
    \Biggl\|
      K_{ji}(\beta_0;
        \widehat e,\widehat G,
        \widehat\Phi^{\mathrm{fit}},\widehat m)
      -
      K_{ji}(\beta_0;\eta_0)
    \Biggr\|^2
    \,\Biggm|\,
    \mathcal T_{fg}
  \Biggr]
  \\
  &=
  o_p(1).
\end{aligned}
\]
Therefore, for every $\varepsilon>0$, conditional Markov's inequality gives
\[
\begin{aligned}
  &\Prb\!\Biggl(
    \frac{1}{|\mathcal B_{fg}|}
    \sum_{(i,j)\in\mathcal B_{fg}}
    \Biggl\|
      \frac{1}{2}
      \Big\{
        K_{ij}(\beta_0;
          \widehat e,\widehat G,
          \widehat\Phi^{\mathrm{fit}},\widehat m)
        +
        K_{ji}(\beta_0;
          \widehat e,\widehat G,
          \widehat\Phi^{\mathrm{fit}},\widehat m)
      \Big\}
      -
      \bar K(O_i,O_j)
    \Biggr\|^2
    >
    \varepsilon
    \,\Biggm|\,
    \mathcal T_{fg}
  \Biggr)
  \\
  &\leq
  \frac{1}{\varepsilon|\mathcal B_{fg}|}
  \sum_{(i,j)\in\mathcal B_{fg}}
  \E\!\Biggl[
    \Biggl\|
      \frac{1}{2}
      \Big\{
        K_{ij}(\beta_0;
          \widehat e,\widehat G,
          \widehat\Phi^{\mathrm{fit}},\widehat m)
        +
        K_{ji}(\beta_0;
          \widehat e,\widehat G,
          \widehat\Phi^{\mathrm{fit}},\widehat m)
      \Big\}
      -
      \bar K(O_i,O_j)
    \Biggr\|^2
    \,\Biggm|\,
    \mathcal T_{fg}
  \Biggr]
  \\
  &=
  o_p(1).
\end{aligned}
\]
Hence, for each $f,g\in\{1,\ldots,F\}$,
\[
  \frac{1}{|\mathcal B_{fg}|}
  \sum_{(i,j)\in\mathcal B_{fg}}
  \Biggl\|
    \frac{1}{2}
    \Big\{
      K_{ij}(\beta_0;
        \widehat e,\widehat G,
        \widehat\Phi^{\mathrm{fit}},\widehat m)
      +
      K_{ji}(\beta_0;
        \widehat e,\widehat G,
        \widehat\Phi^{\mathrm{fit}},\widehat m)
    \Big\}
    -
    \bar K(O_i,O_j)
  \Biggr\|^2
  =
  o_p(1).
\]
Each pair $(i,j)$ with $i\ne j$ appears in exactly one $\mathcal B_{fg}$. Therefore,
\begin{equation}
\label{eq:beta0-kernel-bound}
\begin{aligned}
  &\frac{1}{n(n-1)}
  \sum_{i\ne j}
  \Biggl\|
    \frac{1}{2}
    \Big\{
      K_{ij}(\beta_0;
        \widehat e,\widehat G,
        \widehat\Phi^{\mathrm{fit}},\widehat m)
      +
      K_{ji}(\beta_0;
        \widehat e,\widehat G,
        \widehat\Phi^{\mathrm{fit}},\widehat m)
    \Big\}
    -
    \bar K(O_i,O_j)
  \Biggr\|^2
  \\
  &=
  \sum_{f=1}^F\sum_{g=1}^F
  \frac{|\mathcal B_{fg}|}{n(n-1)}
  \Biggl[
    \frac{1}{|\mathcal B_{fg}|}
    \sum_{(i,j)\in\mathcal B_{fg}}
    \Biggl\|
      \frac{1}{2}
      \Big\{
        K_{ij}(\beta_0;
          \widehat e,\widehat G,
          \widehat\Phi^{\mathrm{fit}},\widehat m)
        +
        K_{ji}(\beta_0;
          \widehat e,\widehat G,
          \widehat\Phi^{\mathrm{fit}},\widehat m)
      \Big\}
      -
      \bar K(O_i,O_j)
    \Biggr\|^2
  \Biggr]
  \\
  &=
  o_p(1).
\end{aligned}
\end{equation}
Combining \eqref{eq:projection-error-bound}, \eqref{eq:beta-perturbation-bound}, and \eqref{eq:beta0-kernel-bound} gives
\[
  \frac{1}{n}\sum_{i=1}^n
  \Biggl\|
    \widehat\kappa_i^{\mathrm{fit}}
    -
    \widehat\kappa_i^\ast
  \Biggr\|^2
  =
  o_p(1).
\]

Recall that
\[
  \widehat\Sigma_\beta
  =
  4\left\{
    \left.
    \frac{\partial}{\partial\beta^\top}
    \widehat U_n(\beta)
    \right|_{\beta=\widehat\beta}
  \right\}^{-1}
  \left\{
    \frac{1}{n}\sum_{i=1}^{n}
    \widehat\kappa_i\widehat\kappa_i^\top
  \right\}
  \left\{
    \left.
    \frac{\partial}{\partial\beta^\top}
    \widehat U_n(\beta)
    \right|_{\beta=\widehat\beta}
  \right\}^{-\top},
\]
where $\widehat\kappa_i$ is the centered row average of $\widehat{\bar K}_{ij}$, the symmetrized estimating-function contribution from pair $(i,j)$:
\[
  \widehat\kappa_i
  =
  \frac{1}{n-1}\sum_{j\ne i}
  \widehat{\bar K}_{ij}
  -
  \frac{1}{n}\sum_{r=1}^{n}
  \frac{1}{n-1}\sum_{s\ne r}
  \widehat{\bar K}_{rs}.
\]
When the fitted future-trajectory distributions have finite supports,
\[
\begin{aligned}
  \widehat\Phi_{ij,\ell}^{\mathrm{det}}(\beta)
  &=
  \sum_{u_1\in\mathcal S_{1,\ell}}
  \sum_{u_2\in\mathcal S_{2,\ell}}
  \Psi_{ij,\ell}(\beta;u_1,u_2)\,
  \widehat p_{1,\ell}(u_1\mid\mathcal H_{i,\ell}^{+})\,
  \widehat p_{2,\ell}(u_2\mid\mathcal H_{j,\ell}^{+})
  \\
  &=
  \int
  \Psi_{ij,\ell}(\beta;u_1,u_2)\,
  d\widehat{\Law}(u_1\mid\mathcal H_{i,\ell}^{+})\,
  d\widehat{\Law}(u_2\mid\mathcal H_{j,\ell}^{+})
  \\
  &=
  \widehat\Phi_{ij,\ell}^{\mathrm{fit}}(\beta).
\end{aligned}
\]
Therefore, under the deterministic implementation,
\[
\begin{aligned}
  \widehat\kappa_i
  &=
  \frac{1}{n-1}\sum_{j\ne i}
  \frac{1}{2}
  \Big\{
    K_{ij}(\widehat\beta;
      \widehat e,\widehat G,\widehat\Phi^{\mathrm{det}},\widehat m)
    +
    K_{ji}(\widehat\beta;
      \widehat e,\widehat G,\widehat\Phi^{\mathrm{det}},\widehat m)
  \Big\}
  \\
  &\quad-
  \frac{1}{n}\sum_{r=1}^{n}
  \frac{1}{n-1}\sum_{s\ne r}
  \frac{1}{2}
  \Big\{
    K_{rs}(\widehat\beta;
      \widehat e,\widehat G,\widehat\Phi^{\mathrm{det}},\widehat m)
    +
    K_{sr}(\widehat\beta;
      \widehat e,\widehat G,\widehat\Phi^{\mathrm{det}},\widehat m)
  \Big\}
  \\
  &=
  \frac{1}{n-1}\sum_{j\ne i}
  \frac{1}{2}
  \Big\{
    K_{ij}(\widehat\beta;
      \widehat e,\widehat G,\widehat\Phi^{\mathrm{fit}},\widehat m)
    +
    K_{ji}(\widehat\beta;
      \widehat e,\widehat G,\widehat\Phi^{\mathrm{fit}},\widehat m)
  \Big\}
  \\
  &\quad-
  \frac{1}{n}\sum_{r=1}^{n}
  \frac{1}{n-1}\sum_{s\ne r}
  \frac{1}{2}
  \Big\{
    K_{rs}(\widehat\beta;
      \widehat e,\widehat G,\widehat\Phi^{\mathrm{fit}},\widehat m)
    +
    K_{sr}(\widehat\beta;
      \widehat e,\widehat G,\widehat\Phi^{\mathrm{fit}},\widehat m)
  \Big\}
  \\
  &=
  \widehat\kappa_i^{\mathrm{fit}}.
\end{aligned}
\]

\noindent\textbf{Alternative Monte Carlo implementation of the future-score projection.}
When $\widehat\Phi^B$ in \eqref{eq:mc-projection} is used to approximate the future-score projection in computing $\widehat\beta$, we compute the symmetrized pair kernels $\widehat{\bar K}_{ij}$ and the resulting empirical Hoeffding projections $\widehat\kappa_i$ for Wald inference using a second, independent evaluation of \eqref{eq:mc-projection}. Condition throughout the following calculation on the observed data, fitted models, and Monte Carlo draws used to compute $\widehat\beta$.

Define the symmetrized Monte Carlo kernel error
\[
\begin{aligned}
  \Delta_{ij}^{\mathrm{MC}}
  &=
  \frac{1}{2}
  \Big[
    K_{ij}(\widehat\beta;
      \widehat e,\widehat G,\widehat\Phi^B,\widehat m)
    -
    K_{ij}(\widehat\beta;
      \widehat e,\widehat G,\widehat\Phi^{\mathrm{fit}},\widehat m)
    \\
  &\hspace{4.8em}+
    K_{ji}(\widehat\beta;
      \widehat e,\widehat G,\widehat\Phi^B,\widehat m)
    -
    K_{ji}(\widehat\beta;
      \widehat e,\widehat G,\widehat\Phi^{\mathrm{fit}},\widehat m)
  \Big].
\end{aligned}
\]
After symmetrization, the conditions of \Cref{lem:mc-error} imply
\[
  \E\!\left(\Delta_{ij}^{\mathrm{MC}}\right)=0,
  \qquad
  \sup_{i\ne j}
  \E\!\left(\|\Delta_{ij}^{\mathrm{MC}}\|^2\right)
  =
  O_p(1),
\]
and for each fixed $i$ the collection $\{\Delta_{ij}^{\mathrm{MC}}:j\ne i\}$ is conditionally independent. Subtracting the corresponding formulas for $\widehat\kappa_i$ and $\widehat\kappa_i^{\mathrm{fit}}$ gives
\[
\begin{aligned}
  \widehat\kappa_i-\widehat\kappa_i^{\mathrm{fit}}
  &=
  \frac{1}{n-1}\sum_{j\ne i}\Delta_{ij}^{\mathrm{MC}}
  -
  \frac{1}{n}\sum_{r=1}^n
  \frac{1}{n-1}\sum_{s\ne r}\Delta_{rs}^{\mathrm{MC}}.
\end{aligned}
\]
Since the second term is the average of the first terms across subjects,
\[
  \frac{1}{n}\sum_{i=1}^n
  \|\widehat\kappa_i-\widehat\kappa_i^{\mathrm{fit}}\|^2
  \le
  \frac{1}{n}\sum_{i=1}^n
  \left\|
    \frac{1}{n-1}\sum_{j\ne i}\Delta_{ij}^{\mathrm{MC}}
  \right\|^2.
\]
For each $i$,
\[
\begin{aligned}
  \E\!\left[
    \left\|
      \frac{1}{n-1}\sum_{j\ne i}\Delta_{ij}^{\mathrm{MC}}
    \right\|^2
  \right]
  &=
  \frac{1}{(n-1)^2}
  \sum_{j\ne i}
  \E\!\left(\|\Delta_{ij}^{\mathrm{MC}}\|^2\right)
  \\
  &\quad+
  \frac{1}{(n-1)^2}
  \sum_{\substack{j\ne i,\ k\ne i\\j\ne k}}
  \E\!\left\{
    (\Delta_{ij}^{\mathrm{MC}})^\top
    \Delta_{ik}^{\mathrm{MC}}
  \right\}
  \\
  &=
  \frac{1}{(n-1)^2}
  \sum_{j\ne i}
  \E\!\left(\|\Delta_{ij}^{\mathrm{MC}}\|^2\right)
  =
  O_p(n^{-1}).
\end{aligned}
\]
The second equality follows because, for fixed $i$, the cross terms involve conditionally independent centered errors. Averaging over $i$ and applying the preceding inequality therefore give
\[
  \E\!\left[
    \frac{1}{n}\sum_{i=1}^n
    \|\widehat\kappa_i-\widehat\kappa_i^{\mathrm{fit}}\|^2
  \right]
  =
  O_p(n^{-1}).
\]
Conditional Markov's inequality then yields
\[
  \frac{1}{n}\sum_{i=1}^n
  \|\widehat\kappa_i-\widehat\kappa_i^{\mathrm{fit}}\|^2
  =
  o_p(1).
\]
Hence
\[
  \frac{1}{n}\sum_{i=1}^n
  \widehat\kappa_i\widehat\kappa_i^\top
  \xrightarrow{p}
  \Var\{\kappa(O)\}.
\]
Together with
\[
  \left.
  \frac{\partial}{\partial\beta^\top}
  \widehat U_n(\beta)
  \right|_{\beta=\widehat\beta}
  \xrightarrow{p}
  \left.
  \frac{\partial}{\partial\beta^\top}
  \psi_\beta(P_0)
  \right|_{\beta=\beta_0},
\]
the continuous mapping theorem yields $\widehat\Sigma_\beta\xrightarrow{p}\Sigma_\beta^{\mathrm{eff}}$. Slutsky's theorem and \Cref{thm:asymptotic-linearity} then establish the stated Wald pivot and coverage.
\end{proof}

\begin{proof}[Proof of \headingcref{cor:semiparametric-efficiency}{Corollary}.]

By \Cref{prop:canonical-gradient}, $D_\beta^{\mathrm{eff}}$ is the canonical gradient of the root functional; by \Cref{thm:asymptotic-linearity}, the estimator has this same influence function. For any regular parametric submodel $\{P_\varepsilon\}$ through $P_0$ with score $s$ and any fixed scalar $t$, $P_{t/\sqrt n}^{\otimes n}$ is contiguous to $P_0^{\otimes n}$. Hence the $o_{P_0^{\otimes n}}(1)$ remainder in \Cref{thm:asymptotic-linearity} is also $o_{P_{t/\sqrt n}^{\otimes n}}(1)$. Under $P_0^{\otimes n}$, the local likelihood-ratio expansion and the multivariate central limit theorem for $D_\beta^{\mathrm{eff}}(O)$ and $s(O)$ imply that the influence-function sum $n^{-1/2}\sum_{i=1}^n D_\beta^{\mathrm{eff}}(O_i)$ and the log-likelihood ratio of $P_{t/\sqrt n}^{\otimes n}$ to $P_0^{\otimes n}$ are jointly asymptotically normal, with limiting cross-covariance $t\E\{D_\beta^{\mathrm{eff}}(O)s(O)\}$. Le Cam's third lemma then yields, under $P_{t/\sqrt n}^{\otimes n}$,
\[
  \sqrt n(\widehat\beta-\beta_0)
  \rightsquigarrow
  N\!\left(
  t\E\{D_\beta^{\mathrm{eff}}(O)s(O)\},
  \Sigma_\beta^{\mathrm{eff}}
  \right),
\]
whereas pathwise differentiability yields
\[
  \sqrt n\{\beta(P_{t/\sqrt n})-\beta_0\}
  \longrightarrow
  t\E\{D_\beta^{\mathrm{eff}}(O)s(O)\}.
\]
The centered limit of $\sqrt n\{\widehat\beta-\beta(P_{t/\sqrt n})\}$ is therefore invariant over such submodels, establishing regularity and variance $\Sigma_\beta^{\mathrm{eff}}=\Var(D_\beta^{\mathrm{eff}})$.

If $D_{\widetilde\beta}$ is the influence function of any other regular asymptotically linear estimator in the same model, the projection property of the canonical gradient yields
\[
  \E\!\left[
  \{D_{\widetilde\beta}-D_\beta^{\mathrm{eff}}\}
  D_\beta^{\mathrm{eff}\top}
  \right]
  =0.
\]
Consequently,
\[
  \Var(D_{\widetilde\beta})
  =
  \Var(D_\beta^{\mathrm{eff}})
  +
  \Var\{D_{\widetilde\beta}-D_\beta^{\mathrm{eff}}\},
\]
which establishes the lower bound and the equality condition.
\end{proof}

\section{Simulation Supplement}
\label{app:simulation-supplement}

\subsection{Simulation Estimators}
\label{app:simulation-estimators}

The four estimators in the simulation tables differ in which components of the observed-data score are included; all four estimators target the same complete-data logit projection in Definition~\ref{def:target} when their required nuisance conditions hold. For $i\ne j$, write
\[
  \omega_{ij}^{1,0}(e)=\frac{A_i(1-A_j)}{e(X_i)\{1-e(X_j)\}},
\]
and write $S_{ij}^{\mathrm{IPW}}(\beta;G)$ and $S_{ij}^{\mathrm{FC}}(\beta;G,\Phi)$ for the scores in Section~\ref{sec:full-correction} evaluated on pair $(i,j)$. Let $m^{\mathrm{IPW}}$ and $m^{\mathrm{FC}}$ denote the outcome regressions used with these two scores. Under the correctly specified censoring law, the conditional mean of $S^{\mathrm{IPW}}$ among treated-versus-control pairs is $m_0$; under the correctly specified censoring law or correctly specified future-score projection, the same holds for $S^{\mathrm{FC}}$.

The four simulation estimators are the roots of
\[
  \widehat U_n^{q}(\beta)
  =
  \frac{1}{n(n-1)}\sum_{i\ne j} K_{ij}^{q}(\beta;\widehat\eta_{ij})=0,
\]
with cross-fitted nuisance estimates, where $q$ takes one of the following forms:
\[
\begin{array}{ll}
\mathrm{IPW}:&
K_{ij}^{\mathrm{IPW}}(\beta)
=
 \omega_{ij}^{1,0}(e)S_{ij}^{\mathrm{IPW}}(\beta;G),\\[0.6em]
\mathrm{IPW\mbox{-}FC}:&
K_{ij}^{\mathrm{IPW\mbox{-}FC}}(\beta)
=
 \omega_{ij}^{1,0}(e)S_{ij}^{\mathrm{FC}}(\beta;G,\Phi),\\[0.6em]
\mathrm{AIPW}:&
K_{ij}^{\mathrm{AIPW}}(\beta)
=
 \omega_{ij}^{1,0}(e)\{S_{ij}^{\mathrm{IPW}}(\beta;G)-m^{\mathrm{IPW}}(X_i,X_j;\beta)\}
 +m^{\mathrm{IPW}}(X_i,X_j;\beta),\\[0.6em]
\mathrm{AIPW\mbox{-}FC}:&
K_{ij}^{\mathrm{AIPW\mbox{-}FC}}(\beta)
=
 \omega_{ij}^{1,0}(e)\{S_{ij}^{\mathrm{FC}}(\beta;G,\Phi)-m^{\mathrm{FC}}(X_i,X_j;\beta)\}
 +m^{\mathrm{FC}}(X_i,X_j;\beta).
\end{array}
\]
IPW and AIPW use the interval score weighted by inverse censoring probabilities. IPW-FC and AIPW-FC add the predictable future-score correction. The augmented estimators, AIPW and AIPW-FC, also include baseline outcome regression. The main text focuses on AIPW-FC, the estimator combining both augmentation components.

\subsection{Data-Generating Mechanism}

The analysis grid was $t_\ell=3\ell$ months, $\ell=0,\ldots,8$. For $i=1,\ldots,n$, let $D_{i,\ell}$ and $H_{i,\ell}$ denote death and hospitalization history by $t_\ell$. Baseline variables were generated as $X_i\sim\mathrm{Uniform}(-1,1)$ and $Z_i^{\mathrm{cvd}}\sim\mathrm{Bernoulli}(0.30)$, independently, followed by
\[
 A_i\mid X_i,Z_i^{\mathrm{cvd}}
 \sim
 \mathrm{Bernoulli}\{\expit(-0.30+0.40X_i+1.60Z_i^{\mathrm{cvd}})\}.
\]
Among subjects alive at the start of interval $\ell$, the discrete death hazard was
\[
 \lambda_{i,\ell}^D
 =
 1-\exp\left[-\exp\{-5.12+0.20X_i+0.45Z_i^{\mathrm{cvd}}-0.05A_i\}\right].
\]
Among subjects alive with no prior hospitalization, the hospitalization hazard was
\[
 \lambda_{i,\ell}^H
 =
 1-\exp\left[
 -\exp\{-1.965+0.45X_i+0.65Z_i^{\mathrm{cvd}}-0.20A_i-1.00A_i Z_i^{\mathrm{cvd}}\}
 \right].
\]
Within each interval, death was simulated before hospitalization; hospitalization was then simulated among subjects who remained alive. Once death occurred, no further outcomes were generated for that subject, whereas a first hospitalization set $H_{i,\ell}=1$ thereafter.

The discrete censoring hazard among uncensored subjects alive through interval $\ell$ was
\[
 \lambda_{i,\ell}^C
 =
 1-\exp\left[
 -\exp\{a_C+0.75H_{i,\ell-1}+0.25X_i+0.55Z_i^{\mathrm{cvd}}\}
 \right].
\]
Events at $t_\ell$ were ascertained before censoring at $t_\ell$, in agreement with the interval convention used in Section~\ref{sec:complete-data-target}. For each scenario, only the intercept $a_C$ of this censoring hazard was calibrated so that the marginal censoring proportion by 24 months was approximately the targeted 30\%, 50\%, or 65\%; the covariate effects in the censoring model and the event-generating mechanism were otherwise unchanged. Across replications, the realized censoring proportions were essentially at their targets, and the mortality and hospitalization rates remained approximately 5\% and 60\%, respectively.

Pairwise win and determinacy processes followed the two-level priority rule that ranks death before hospitalization, as described in Section~\ref{sec:simulation-studies}. The pair design vector was
\[
 \widetilde X_{ij}
 =
 (1, X_i-X_j, Z_i^{\mathrm{cvd}}-Z_j^{\mathrm{cvd}})^\top.
\]

The correctly specified treatment model was logistic in $(X,Z^{\mathrm{cvd}})$. The censoring, death, and hospitalization hazards were fitted by pooled complementary log-log models using the covariates in their respective data-generating hazards. We used five subject-level folds throughout. For a pair whose members belonged to folds $f$ and $g$, the transition fits and the outcome-regression fits excluded both folds. The final estimating equations included all ordered pairs.

The future state was $(D_{i,\ell},H_{i,\ell})\in\{0,1\}^2$. Fitted one-step death and hospitalization hazards were used to update state probabilities recursively under $A=1$ and $A=0$, and the two subject distributions were then combined to obtain expected future win and resolution contributions. The same finite-state recursion supplied both the future-score projection and the baseline outcome regression for AIPW and AIPW-FC. For bootstrap inference, subjects were sampled with replacement, all nuisance models and pair indices were reconstructed, and each estimating equation was solved. Basic bootstrap intervals were formed by reflecting the 0.025 and 0.975 bootstrap quantiles about the original estimate.

The benchmark value of the complete-data target was calculated independently of the observed censoring process. We used deterministic integration over the continuous biomarker distribution, probabilities 0.7 and 0.3 for the two values of $Z^{\mathrm{cvd}}$, and deterministic forward recursion under $A=1$ and $A=0$. The resulting weighted population logistic score was solved for $(\alpha,\gamma_X,\gamma_{\mathrm{cvd}})$. This procedure yielded the common benchmark value $(0.583789,-0.538483,-0.273921)^\top$ for all three censoring scenarios.

\subsection{Double-Robustness Study}
\label{app:simulation-dr}

The double-robustness study focused on the 50\% censoring scenario and examined the two robustness components in Theorem~\ref{thm:aipw-dr}. Four misspecification scenarios were considered. For treatment robustness, the misspecified propensity model omitted $Z^{\mathrm{cvd}}$ while the outcome regression $m$ was correctly specified, or the propensity model was correct while $m$ was set to zero. For censoring robustness, the misspecified censoring model contained only the interval offset and omitted $(H_{\ell-1},X,Z^{\mathrm{cvd}})$ while the future-score projection was correct, or the censoring model was correct while the death and hospitalization transition models used for $\Phi$ omitted $Z^{\mathrm{cvd}}$ and the $A Z^{\mathrm{cvd}}$ interaction. All remaining nuisance models were correctly specified. Each scenario used 500 replications. For this auxiliary study, the intercept of the same censoring hazard was calibrated in the same manner to maintain approximately 50\% marginal censoring; all other data-generating parameters were identical. The model-based variance estimator was implemented only for AIPW-FC, and its ASE and coverage under these misspecification scenarios are reported as an empirical assessment.

\begin{landscape}
\small
\setlength{\tabcolsep}{5pt}
\begin{longtable}{lllrrrrr}
\caption{Treatment double-robustness results under 50\% censoring.}
\label{tab:simulation-treatment-dr}\\
\toprule
Scenario & Component & Estimator & True & Bias & ESD & ASE & Coverage \\
\midrule
\endfirsthead
\multicolumn{8}{c}{\tablename\ \thetable\ (continued)}\\
\toprule
Scenario & Component & Estimator & True & Bias & ESD & ASE & Coverage \\
\midrule
\endhead
\bottomrule
\endlastfoot
$e$ wrong, $m$ correct & Intercept & IPW & 0.58 & -0.08 & 0.19 & -- & -- \\
                       & & AIPW & 0.58 & 0.00 & 0.20 & -- & -- \\
                       & & IPW-FC & 0.58 & -0.08 & 0.18 & -- & -- \\
                       & & AIPW-FC & 0.58 & 0.00 & 0.18 & 0.17 & 0.94 \\
\addlinespace
                       & Biomarker & IPW & -0.54 & 0.00 & 0.16 & -- & -- \\
                       & & AIPW & -0.54 & 0.00 & 0.17 & -- & -- \\
                       & & IPW-FC & -0.54 & 0.00 & 0.15 & -- & -- \\
                       & & AIPW-FC & -0.54 & 0.00 & 0.15 & 0.15 & 0.95 \\
\addlinespace
                       & CVD & IPW & -0.27 & 0.20 & 0.22 & -- & -- \\
                       & & AIPW & -0.27 & -0.02 & 0.24 & -- & -- \\
                       & & IPW-FC & -0.27 & 0.21 & 0.20 & -- & -- \\
                       & & AIPW-FC & -0.27 & -0.01 & 0.21 & 0.22 & 0.95 \\
\midrule
$e$ correct, $m$ wrong & Intercept & IPW & 0.58 & 0.00 & 0.20 & -- & -- \\
                       & & AIPW & 0.58 & 0.00 & 0.20 & -- & -- \\
                       & & IPW-FC & 0.58 & 0.00 & 0.18 & -- & -- \\
                       & & AIPW-FC & 0.58 & 0.00 & 0.18 & 0.18 & 0.95 \\
\addlinespace
                       & Biomarker & IPW & -0.54 & 0.00 & 0.18 & -- & -- \\
                       & & AIPW & -0.54 & 0.00 & 0.18 & -- & -- \\
                       & & IPW-FC & -0.54 & 0.00 & 0.16 & -- & -- \\
                       & & AIPW-FC & -0.54 & 0.00 & 0.16 & 0.16 & 0.94 \\
\addlinespace
                       & CVD & IPW & -0.27 & -0.02 & 0.23 & -- & -- \\
                       & & AIPW & -0.27 & -0.02 & 0.23 & -- & -- \\
                       & & IPW-FC & -0.27 & -0.01 & 0.21 & -- & -- \\
                       & & AIPW-FC & -0.27 & -0.01 & 0.21 & 0.20 & 0.94 \\
\end{longtable}
\end{landscape}

\begin{landscape}
\small
\setlength{\tabcolsep}{5pt}
\begin{longtable}{lllrrrrr}
\caption{Censoring double-robustness results under 50\% censoring.}
\label{tab:simulation-censoring-dr}\\
\toprule
Scenario & Component & Estimator & True & Bias & ESD & ASE & Coverage \\
\midrule
\endfirsthead
\multicolumn{8}{c}{\tablename\ \thetable\ (continued)}\\
\toprule
Scenario & Component & Estimator & True & Bias & ESD & ASE & Coverage \\
\midrule
\endhead
\bottomrule
\endlastfoot
$G$ wrong, $\Phi$ correct & Intercept & IPW & 0.58 & -0.06 & 0.19 & -- & -- \\
                          & & AIPW & 0.58 & -0.07 & 0.20 & -- & -- \\
                          & & IPW-FC & 0.58 & 0.00 & 0.18 & -- & -- \\
                          & & AIPW-FC & 0.58 & 0.00 & 0.18 & 0.19 & 0.96 \\
\addlinespace
                          & Biomarker & IPW & -0.54 & 0.02 & 0.17 & -- & -- \\
                          & & AIPW & -0.54 & 0.03 & 0.18 & -- & -- \\
                          & & IPW-FC & -0.54 & -0.01 & 0.16 & -- & -- \\
                          & & AIPW-FC & -0.54 & 0.00 & 0.16 & 0.17 & 0.96 \\
\addlinespace
                          & CVD & IPW & -0.27 & 0.03 & 0.22 & -- & -- \\
                          & & AIPW & -0.27 & 0.04 & 0.23 & -- & -- \\
                          & & IPW-FC & -0.27 & -0.01 & 0.21 & -- & -- \\
                          & & AIPW-FC & -0.27 & 0.00 & 0.21 & 0.21 & 0.94 \\
\midrule
$G$ correct, $\Phi$ wrong & Intercept & IPW & 0.58 & 0.00 & 0.20 & -- & -- \\
                          & & AIPW & 0.58 & 0.00 & 0.20 & -- & -- \\
                          & & IPW-FC & 0.58 & 0.01 & 0.18 & -- & -- \\
                          & & AIPW-FC & 0.58 & 0.00 & 0.18 & 0.19 & 0.95 \\
\addlinespace
                          & Biomarker & IPW & -0.54 & 0.00 & 0.18 & -- & -- \\
                          & & AIPW & -0.54 & 0.00 & 0.18 & -- & -- \\
                          & & IPW-FC & -0.54 & 0.00 & 0.16 & -- & -- \\
                          & & AIPW-FC & -0.54 & 0.00 & 0.16 & 0.17 & 0.95 \\
\addlinespace
                          & CVD & IPW & -0.27 & -0.02 & 0.23 & -- & -- \\
                          & & AIPW & -0.27 & -0.01 & 0.23 & -- & -- \\
                          & & IPW-FC & -0.27 & -0.01 & 0.21 & -- & -- \\
                          & & AIPW-FC & -0.27 & 0.00 & 0.21 & 0.21 & 0.96 \\
\end{longtable}
\end{landscape}

The treatment misspecification scenarios evaluate robustness of the treatment adjustment. When $e$ omitted the strong baseline CVD predictor, IPW and IPW-FC had biases of about $-0.08$ for the intercept and 0.20 for the CVD coefficient, whereas both augmented estimators with correct $m$ were close to the target value. When $e$ was correct and $m$ was set to zero, all four estimators remained close to the target value. The censoring misspecification scenarios evaluate robustness of the censoring adjustment. When $G$ omitted all history and baseline predictors, IPW and AIPW were biased in all three components, while IPW-FC and AIPW-FC with correct $\Phi$ had much smaller bias. When $G$ was correct and $\Phi$ was misspecified, all estimators again remained close to the target value. Across the four scenarios, model-based coverage for AIPW-FC ranged from 0.94 to 0.96.

\section{OneFlorida Application Details}
\label{app:oneflorida-implementation}

\subsection{Cohort and Outcome}

The index date was the date of breast cancer surgery. The analysis required at least one recorded encounter in each of the years before and after the index date, pre-index cardiovascular disease, membership in the high-risk frailty group, and no other cancer recorded in the preceding year. Patients were also required to be observed and alive at the 90-day landmark and to have nonmissing age and frailty measurements. Initiation of adjuvant chemotherapy between surgery and the landmark defined $A=1$; absence of initiation during that window defined $A=0$.

Follow-up after the landmark was divided into 36 complete intervals of 30 days; the one- and two-year analyses truncated follow-up after 12 and 24 complete intervals, respectively. Let $D_{i,\ell}$ denote death status and $H_{i,\ell}$ indicate whether patient $i$ had experienced an inpatient hospitalization before interval $\ell$. For an ordered pair $(i,j)$, the interval comparison was determined by
\[
R_{ij,\ell}
=
\ind(D_{i,\ell}\ne D_{j,\ell})
+
\ind(D_{i,\ell}=D_{j,\ell})
\ind(H_{i,\ell}\ne H_{j,\ell}),
\]
with $W_{ij,\ell}=1$ when the first member was alive while the second had died, or, conditional on equal death status, when the first member had no hospitalization history while the second did.

The pair design vector was
\[
\widetilde X_{ij}
=
\left(
1,\,
X^{F}_i-X^{F}_j
\right)^\top,
\]
where $X^{F}$ was the standardized frailty index. The first component of $\beta$ is therefore the treated-versus-control coefficient. The analysis contained 238,161 treated-versus-control ordered pairs and 1,591,382 ordered pairs in total.

\subsection{Nuisance Models and Future-Score Recursion}

All nuisance models were fitted with five-fold cross-fitting at the subject level. The propensity score model was logistic in standardized age, standardized frailty, standardized log prior-year encounter utilization, indicators for congestive heart failure, pre-index cardiomyopathy, pre-index atrial fibrillation or other arrhythmia/cardiac-arrest diagnosis, diabetes, renal disease, myocardial infarction, and stroke or transient ischemic attack, and race and ethnicity group and payer group. The discrete censoring hazard was modeled by a pooled complementary log-log regression on interval, treatment, age, frailty, hospitalization history, prior-year encounter utilization, log-transformed histories of marker abnormality and acute-care use available at the start of the interval, and race and ethnicity group and payer group.

For the future-score projection, the recursion was based on the state $(D_{i,\ell},H_{i,\ell})\in\{0,1\}^2$. One-step death and hospitalization probabilities were estimated by pooled complementary log-log models. Both models included interval, treatment, age, frailty, and landmark summaries of marker abnormality and acute-care utilization; the death model additionally included current hospitalization history. The hospitalization model was fitted among patients alive and free of prior hospitalization. In the recursion, death was evaluated first, and hospitalization was updated only for patients who survived through the interval. Once death occurred, no further state updates were made.

For a patient with covariate history summarized by the fitted transition models, let $\pi_{i,r}^{a}(d,h)$ denote the predicted probability of state $(D,H)=(d,h)$ at the start of interval $r$ under treatment value $a$. With fitted one-step hazards $\widehat p^D_{i,r}(a,h)$ and $\widehat p^H_{i,r}(a,0)$, the application followed the forward recursion
\[
\begin{aligned}
\pi_{i,r+1}^{a}(1,h)
&=
\pi_{i,r}^{a}(1,h)
+
\pi_{i,r}^{a}(0,h)\widehat p^D_{i,r}(a,h),\qquad h=0,1,\\
\pi_{i,r+1}^{a}(0,1)
&=
\pi_{i,r}^{a}(0,1)\{1-\widehat p^D_{i,r}(a,1)\}
+
\pi_{i,r}^{a}(0,0)\{1-\widehat p^D_{i,r}(a,0)\}
\widehat p^H_{i,r}(a,0),\\
\pi_{i,r+1}^{a}(0,0)
&=
\pi_{i,r}^{a}(0,0)\{1-\widehat p^D_{i,r}(a,0)\}
\{1-\widehat p^H_{i,r}(a,0)\}.
\end{aligned}
\]
Let $\mathcal W$ and $\mathcal R$ be the $4\times4$ matrices obtained by applying the same priority rule, ranking death before hospitalization, to two states. Their entries equal one when the ordered pair is a win or is resolved, respectively. For a pair still observed at censoring interval $\ell$, the components of the future win and resolution scores that do not depend on $\beta$ were
\[
\widehat{FW}_{ij,\ell}
=
\sum_{r=\ell+1}^{M}\Delta t_r
\{\widehat\pi_{i,r}^{A_i}\}^\top \mathcal W
\widehat\pi_{j,r}^{A_j},
\qquad
\widehat{FR}_{ij,\ell}
=
\sum_{r=\ell+1}^{M}\Delta t_r
\{\widehat\pi_{i,r}^{A_i}\}^\top \mathcal R
\widehat\pi_{j,r}^{A_j}.
\]
These future components determine the future-score projection through
\[
\widehat\Phi_{ij,\ell}(\beta)
=
\widetilde X_{ij}
\left[
\widehat{FW}_{ij,\ell}
-
\widehat{FR}_{ij,\ell}
\expit\{\beta^\top\widetilde X_{ij}\}
\right].
\]
This finite-state recursion implements \Cref{eq:recursive-projection} deterministically.

The baseline outcome regression used the same recursion over four states under $A=1$ and $A=0$. Its death and hospitalization transition models included interval, treatment, age, frailty, and prior-year encounter utilization, with hospitalization history additionally included in the death model. The resulting win and resolution components, which are independent of $\beta$, were averaged over all ordered pairs in the AIPW score.

For a pair whose members belonged to folds $f$ and $g$, transition and outcome-regression predictions for the pair came from models fitted outside both folds; subject-level treatment and censoring predictions were likewise evaluated out of fold. Each of the 200 nonparametric bootstrap samples resampled patients with replacement, reconstructed all pair indices, refitted every nuisance model, and solved each estimating equation. Basic bootstrap intervals were calculated as
\[
\left[
2\widehat\beta-
q_{0.975}(\widehat\beta^{\ast}),\,
2\widehat\beta-
q_{0.025}(\widehat\beta^{\ast})
\right].
\]

\section{Sensitivity Analyses for the OneFlorida Application}
\label{app:oneflorida-sensitivity}

We examined three structured changes to the nuisance models at the three-year horizon, one component at a time. These analyses assessed sensitivity to the treatment-assignment model, the censoring model, and the future-score transition model. The reduced propensity model omitted cardiomyopathy, race and ethnicity group, and payer group from the primary propensity model. The reduced censoring model omitted post-landmark hospitalization, marker, and acute-care histories from the censoring model. The reduced future-score transition model omitted marker abnormality and acute-care histories.

\begin{table}[!htbp]
\centering
\caption{OneFlorida analyses under alternative model specifications at the three-year horizon.
The primary analysis is included for reference. Each reduced analysis changes
only the indicated model component. Entries are coefficient estimates
with 200 subject-bootstrap standard errors in parentheses. Relative efficiency
(RE) is the ratio of the bootstrap variance without future-score correction
to that with future-score correction for the treated-versus-control
coefficient.}
\label{tab:oneflorida-sensitivity}
\small
\setlength{\tabcolsep}{5pt}
\begin{tabular}{llccc}
\toprule
Model specification
& Estimator
& Treatment
& Frailty
& RE \\
\midrule
Primary analysis & IPW
& $0.10\ (0.21)$
& $-0.45\ (0.12)$
& -- \\
& IPW-FC
& $0.12\ (0.17)$
& $-0.49\ (0.10)$
& $1.47$ \\
& AIPW
& $0.13\ (0.21)$
& $-0.44\ (0.13)$
& -- \\
& AIPW-FC
& $0.13\ (0.18)$
& $-0.48\ (0.12)$
& $1.48$ \\
\midrule
Reduced propensity model & IPW
& $0.11\ (0.21)$
& $-0.45\ (0.11)$
& -- \\
& IPW-FC
& $0.13\ (0.16)$
& $-0.49\ (0.09)$
& $1.59$ \\
& AIPW
& $0.14\ (0.20)$
& $-0.44\ (0.11)$
& -- \\
& AIPW-FC
& $0.15\ (0.16)$
& $-0.47\ (0.10)$
& $1.51$ \\
\midrule
Reduced censoring model & IPW
& $0.14\ (0.21)$
& $-0.44\ (0.11)$
& -- \\
& IPW-FC
& $0.12\ (0.17)$
& $-0.50\ (0.11)$
& $1.43$ \\
& AIPW
& $0.17\ (0.21)$
& $-0.43\ (0.12)$
& -- \\
& AIPW-FC
& $0.13\ (0.18)$
& $-0.49\ (0.12)$
& $1.43$ \\
\midrule
Reduced future-score transition model & IPW
& $0.10\ (0.21)$
& $-0.45\ (0.12)$
& -- \\
& IPW-FC
& $0.15\ (0.17)$
& $-0.49\ (0.10)$
& $1.44$ \\
& AIPW
& $0.13\ (0.21)$
& $-0.44\ (0.13)$
& -- \\
& AIPW-FC
& $0.16\ (0.18)$
& $-0.48\ (0.12)$
& $1.45$ \\
\bottomrule
\end{tabular}
\end{table}

Across all three reduced specifications, treatment estimates remained close to the primary three-year estimate, frailty remained negative, and future-score correction retained comparable precision gains. For AIPW-FC, treatment estimates under the reduced specifications ranged from 0.13 to 0.16, compared with 0.13 in the primary analysis, and frailty estimates ranged from $-0.49$ to $-0.47$, compared with $-0.48$. Every corresponding 95\% bootstrap confidence interval for treatment included zero. Relative efficiency ranged from 1.43 to 1.59 for IPW-FC versus IPW and from 1.43 to 1.51 for AIPW-FC versus AIPW.

\bibliographystyle{plainnat}
\bibliography{References}
\end{document}